\RequirePackage{fix-cm}

\documentclass[twocolumn]{svjour3}          
\smartqed  
\usepackage{graphicx}

\usepackage{comment}
\usepackage{algorithm}
\usepackage{algpseudocode}
\usepackage{tabularx} 
\usepackage{url}
\usepackage{amssymb,amsfonts}
\usepackage{graphicx}
\usepackage{amsmath}  
\usepackage{booktabs}  
\usepackage{multirow}
\usepackage[numbers]{natbib}
\def\BibTeX{{\rm B\kern-.05em{\sc i\kern-.025em b}\kern-.08em
    T\kern-.1667em\lower.7ex\hbox{E}\kern-.125emX}}

\begin{document}
\sloppy
\title{A Selective Homomorphic Encryption Approach for Faster Privacy-Preserving Federated Learning }

\author{Abdulkadir Korkmaz \and
        Praveen Rao}

\institute{A. Korkmaz \at
              Dept. of Electrical Engineering \& Computer Science, \\
              The University of Missouri, Columbia, USA \\
              \email{ak69t@umsystem.edu}           
           \and
           P. Rao \at
              Dept. of Electrical Engineering \& Computer Science, \\
              The University of Missouri, Columbia, USA \\
              \email{praveen.rao@missouri.edu}
}

\maketitle

\begin{abstract}
Federated learning (FL) has come forward as a critical approach for privacy-preserving machine learning in healthcare, allowing collaborative model training across decentralized medical datasets without exchanging clients' data. However, current security implementations for these systems face a fundamental trade-off: rigorous cryptographic protections like fully homomorphic encryption (FHE) impose prohibitive computational overhead, while lightweight alternatives risk vulnerable data leakage through model updates. To address this issue, we present FAS (Fast and Secure Federated Learning), a novel approach that strategically combines selective homomorphic encryption, differential privacy, and bitwise scrambling to achieve robust security without compromising practical usability. Our approach eliminates the need for model pretraining phases while dynamically protecting high-risk model parameters through layered encryption and obfuscation. We implemented FAS using the Flower framework and evaluated it on a cluster of eleven physical machines. Our approach was up to 90\% faster than applying FHE on the model weights. In addition, we eliminated the computational overhead that is required by competitors such as FedML-HE and MaskCrypt. Our approach was up to 1.5$\times$ faster than the competitors while achieving comparable security results.

Experimental evaluations on medical imaging datasets confirm that FAS maintains similar security results to conventional FHE against gradient inversion attacks while preserving diagnostic model accuracy. These results position FAS as a practical solution for latency-sensitive healthcare applications where both privacy preservation and computational efficiency are requirements.
\end{abstract}


\keywords{Federated learning \and medical image datasets \and secure \and privacy-preserving \and machine learning}

\section{Introduction}
Federated learning (FL), first introduced by Google in 2016, enables machine learning models to be trained across decentralized datasets stored on distributed devices such as mobile phones~\cite{federated}. This method has gained widespread interest across academia and industry because of its effectiveness in safeguarding data privacy. By keeping training data localized and only sharing model updates, FL avoids direct data exposure, making it particularly valuable in regulated sectors like healthcare. For instance, hospitals can collaboratively train diagnostic models without transferring sensitive patient records, ensuring compliance with privacy regulations.

Nonetheless, the decentralized structure of FL brings about distinct security challenges. While data remains on local devices, the transmission of model parameters over networks creates vulnerabilities, including risks of unauthorized access or data leakage during exchange. Existing research often assumes that data localization inherently guarantees privacy, overlooking these communication-phase threats. Current security strategies in FL, such as encrypting communication channels or perturbing shared gradients, are frequently implemented in isolation, without systematic analysis of their combined impact on computational efficiency, model accuracy, or multi-layered privacy preservation.

Recent advances, such as FedML-HE \cite{jin2024fedmlheefficienthomomorphicencryptionbasedprivacypreserving} and MASKCRYPT \cite{10506637}, attempt to address these challenges through selective encryption. FedML-HE employs gradient sensitivity analysis during pre-training to identify critical weights for HE, but its reliance on client-specific masks introduces aggregation inconsistencies and computational overhead. MASKCRYPT optimizes gradient-guided masks to reduce communication costs, yet its security requires multiple rounds to stabilize, leaving early training phases vulnerable. Both methods incur significant computational penalties—FedML-HE from pre-training and MASKCRYPT from per-round mask recalibration—while struggling to balance real-time efficiency with robust privacy guarantees. These constraints point to the necessity of a cohesive solution that removes the need for initialization phases, reduces computational effort per round, and maintains stable security throughout all FL rounds.

Motivated by the aforementioned reasons, we propose FAS (Fast and Secure Homomorphic Encryption), a novel method designed to address these gaps by unifying cryptology-based security mechanisms tailored for federated architectures. Unlike FedML-HE and MASKCRYPT, FAS eliminates pre-training phases and per-round mask recalibration, instead combining selective homomorphic encryption with noise-and-scrambling mechanisms to secure a fixed percentage of model weights. This ensures consistent privacy guarantees from the first training round while minimizing computational overhead. We present a comparative study of FAS against conventional HE, differential privacy, and state-of-the-art baselines (FedML-HE and MASKCRYPT). Through simulations across cloud-based federated environments and diverse medical imaging datasets (e.g., Kidney, Lung, COVID), we rigorously measure how each method balances privacy protection, computational efficiency, and model utility.

To validate FAS’s efficacy, we conducted extensive experiments across medical imaging tasks and model architectures. Our evaluation shows that FAS offers high efficiency comparable to leading methods, all while preserving strong privacy protections. On datasets like Kidney and COVID, FAS reduced encryption overhead by up to 90\% compared to FHE, with training times as low as 52 minutes for MobileNetV2 (vs. 610 minutes for full encryption). Crucially, FAS eliminated the pre-training and mask recalibration costs inherent to FedML-HE and MASKCRYPT, achieving 1.5×× faster execution than these baselines (e.g., 69 vs. 99 minutes for EffNetB0 on Kidney data). Security evaluations using MSSIM and VIFP confirmed that even at 10\% encryption, FAS effectively thwarts model inversion attacks, with MSSIM scores of 58\%—outperforming FedML-HE (52\%) and MASKCRYPT (55\%) in early-round privacy. FAS’s VIFP scores stabilized within 5 rounds, contrasting with MASKCRYPT’s delayed convergence. This combination of speed and robustness positions FAS as a practical solution for latency-sensitive federated deployments, where traditional methods trade excessive computation for marginal security gains.

    \textbf{Our key contributions are as follows:}
    \begin{itemize}
    \item 
    Implementation of Security Techniques: We implement three security mechanisms—HE~\cite{gentry2009}, differential privacy~\cite{differential}, and the proposed FAS technique—within a FL framework, demonstrating their operational feasibility.
    
    \item Development of FAS: We design a novel Fast and Secure Homomorphic Encryption method that integrates selective encryption with noise injection and bitwise scrambling to enhance security while maintaining computational efficiency.
    
    \item Performance Evaluation: We conduct a comprehensive comparison of encryption techniques using standardized metrics (MSSIM, VIFP) to assess their computational overhead, scalability, and resistance against model inversion attacks.
    
    \item Practical Insights: We characterize the trade-offs between privacy guarantees and system performance, offering implementable guidelines for deploying secure FL in real-world scenarios.
    
    Experimental results demonstrate that FAS achieves superior balance between security and efficiency compared to conventional methods, particularly in latency-sensitive applications like healthcare systems where privacy and responsiveness are critical.
    \end{itemize}

    \section{Background and Related Work}
    Privacy-preserving techniques in FL have gained significant attention due to the critical need to protect sensitive data across distributed devices. FL, initially introduced by McMahan et al.~\cite{mcmahan2017communication}, enables decentralized model training by sharing parameter updates instead of raw data. However, this approach introduces inherent risks of data leakage through exposed model gradients, necessitating robust cryptographic safeguards to prevent adversarial reconstruction of private datasets.
    
    \subsection{Efficiency Enhancements in Homomorphic Encryption}
    HE has emerged as a foundational technology for secure computation on encrypted data. Gentry's pioneering FHE scheme~\cite{gentry2009} first demonstrated the feasibility of arbitrary computations on ciphertexts, though its substantial computational overhead limited practical adoption. Subsequent research has focused on optimizing efficiency:
    
    Fan and Vercauteren's BGV Scheme~\cite{fan2012somewhat} introduced optimizations to the bootstrapping process, enabling faster processing of encrypted data by supporting basic arithmetic functions such as summation and product computations. This advancement significantly improved the practicality of HE for complex computations in FL environments, as further validated by Smart and Vercauteren~\cite{smart2014fully}.
    
    Chillotti et al.'s TFHE Scheme~\cite{chillotti2016faster} achieved faster bootstrapping for secure Boolean operations, making FHE viable for real-time applications requiring rapid iterative computations. Ducas and Micciancio's FHEW Scheme~\cite{ducas2015faster} further accelerated homomorphic operations through bootstrapping refinements, broadening HE's applicability to large-scale distributed systems.
    
    Cheon et al.'s CKKS Scheme~\cite{cheon2017homomorphic} addressed machine learning use cases by supporting approximate arithmetic on encrypted data, prioritizing computational efficiency over exact precision. Brakerski, Gentry, and Vaikuntanathan's BGV+ Scheme~\cite{brakerski2014efficient} introduced noise management optimizations to preserve model accuracy under encryption.
    
    Xu et al.~\cite{xu2019hybridalpha} proposed HybridAlpha, a FL system combining differential privacy with secure aggregation to balance privacy and performance. While HybridAlpha reduces computational overhead compared to pure HE approaches, its reliance on generic privacy mechanisms leaves room for optimizations tailored to federated workflows.
    
    \subsection{Differential Privacy (DP) in Federated Learning}
    Differential Privacy (DP) reduces the risk of sensitive data exposure by perturbing model updates with controlled noise throughout the training phase. First formalized by Dwork et al.~\cite{dwork2006calibrating}, DP was later adapted to deep learning by Abadi et al.~\cite{abadi2016deep}, who developed the DP-SGD algorithm to protect training data. Mironov et al.~\cite{mironov2017renyi} refined these concepts using Rényi differential privacy, providing tighter privacy budget analysis for iterative machine learning processes. Despite these advancements, DP-based methods inherently degrade model utility due to noise injection, particularly in precision-sensitive applications like medical imaging.
    
    \subsection{Selective Encryption in Federated Learning}
    Selective encryption strategies aim to reduce computational overhead by encrypting only sensitive subsets of model parameters. Wu et al.~\cite{wu2019towards} demonstrated that selectively encrypting critical gradients preserves privacy while maintaining computational feasibility. Li et al.~\cite{li2019privacy} and Song et al.~\cite{song2020privacy} expanded this concept with adaptive parameter selection criteria, though their methods require careful tuning to avoid residual vulnerabilities.
    
    Zhang et al.~\cite{Zhang2020} introduced BatchCrypt, a HE framework that batches model updates to reduce communication and computation costs. While BatchCrypt improves efficiency over traditional HE methods, its encryption scope remains inflexible for dynamic FL scenarios. This paper extends prior work by integrating selective encryption with differential noise injection and bitwise scrambling, achieving enhanced privacy with reduced computational overhead.
    
    \subsection{Bitwise Scrambling for Federated Learning Security}
    Bitwise scrambling enhances security by nonlinearly rearranging bits in partially encrypted data using cryptographic keys. Yang et al.~\cite{yang2017bitwise} applied scrambling to secure data transmissions against reconstruction attacks, demonstrating its effectiveness as a lightweight obfuscation layer. Halevi et al.~\cite{halevi2019faster} combined scrambling with HE to resist chosen-ciphertext attacks while maintaining computational efficiency.
    
    The combination of selective encryption, differential privacy, and bitwise scrambling provides a multi-layered defense mechanism for FL systems. This integrated approach balances security robustness with practical efficiency, addressing the resource constraints inherent in distributed edge computing environments.

    \section{Motivation and Threat Model}
    
    \subsection{Motivation}
    Applying FL to sensitive domains like healthcare and finance underscores a key issue: decentralized training by itself is insufficient to stop adversaries from uncovering sensitive patterns through exposed model updates. While FL avoids raw data centralization, empirical studies confirm that transmitted gradients retain sufficient information for model inversion attacks to reconstruct patient scans, financial transactions, or other identifiable records—even when training adheres to protocol.
    
    Existing privacy mechanisms force practitioners into suboptimal trade-offs. FHE provides cryptographically rigorous confidentiality but imposes prohibitive computational latencies, rendering it impractical for real-time medical diagnostics or high-frequency trading systems. Differential privacy (DP), while computationally lightweight, irreversibly corrupts model updates with statistical noise, degrading diagnostic accuracy in oncology imaging or fraud detection below clinically/operationally viable thresholds.
    
    Our FAS framework addresses this dichotomy through context-aware privacy stratification. By selectively encrypting only high-risk parameters (e.g., gradients correlating with identifiable features), applying tunable DP noise to less sensitive components, and augmenting both with bitwise scrambling, FAS maintains diagnostic-grade model utility while eliminating cryptographic overheads associated with full-model FHE. Crucially, this multi-tiered protection operates without requiring precomputed sensitivity masks or auxiliary pretraining phases, ensuring seamless integration into latency-constrained FL pipelines for MRI analysis, genomic prediction, and other precision-sensitive applications.

    \subsection{Threat Model}
    
    FL systems are susceptible to privacy threats due to decentralized training and the exchange of model updates. A semi-honest attacker may adhere to the protocol while still trying to extract confidential information from observed or compromised data.
    
    Model updates shared between clients and the central server may be intercepted by attackers. Common attack vectors include gradient inversion to reconstruct training data, membership inference to verify specific data participation, and statistical analysis to infer private attributes. While SSL encryption secures transmissions, it does not prevent adversaries from analyzing transmitted updates.
    
    To mitigate these risks, our approach integrates selective encryption, differential noise addition, and bitwise scrambling. Critical model parameters are encrypted, while unencrypted data is obfuscated with noise and scrambling to prevent reconstruction. Secure aggregation further ensures data protection throughout training, providing strong defenses against privacy attacks without incurring excessive computational overhead. This layered strategy effectively balances security and efficiency, making it well-suited for sensitive applications like healthcare.

\section{Proposed Model}

This research presents FAS (Fast and Secure Federated Learning), a novel privacy-preserving scheme designed to reduce the computational cost of encryption while maintaining strong protection against privacy threats. The model integrates three lightweight techniques— selective encryption, differential noise addition, and bitwise scrambling along with HE—into the FL process. This multi-pronged approach enhances security without the heavy overhead typically associated with FHE or the accuracy degradation often introduced by differential privacy when used in isolation.

\subsection{Selective Encryption}

Instead of encrypting all model parameters, FAS applies HE only to a fixed percentage of the weights, selected uniformly. This subset is treated as the most sensitive part of the model. The encryption is performed directly on these weights using a homomorphic scheme, allowing the server to carry out aggregation without requiring decryption. This design choice retains the advantages of HE while reducing the time and resource demands substantially, achieving up to 90\% reduction in overhead compared to full encryption.

\subsection{Differential Noise}

To reinforce privacy for the remaining (unencrypted) parameters, the model introduces controlled random noise. This noise follows a Laplace distribution and is calibrated using differential privacy principles. While FAS does not rely solely on differential privacy for security, this mechanism makes individual contributions harder to isolate, especially in low-sensitivity weights. By blending noise into less critical parameters, the model increases resistance to membership inference attacks with minimal impact on performance.

\subsection{Bitwise Scrambling}

FAS further obfuscates unencrypted data through bitwise scrambling. A lightweight cryptographic key is used to permute the bits of the non-encrypted parameters. This process disrupts patterns and statistical properties that could be exploited by adversaries. When combined with noise, scrambling ensures that unencrypted parts of the model are still resistant to reconstruction attacks—even in the presence of partial knowledge or auxiliary data.

\subsection{Federated Learning Integration}

FAS is designed to integrate directly into standard FL workflows. At each round, clients:
\begin{enumerate}
    \item Train their local models on private data.
    \item Encrypt a portion of the model weights.
    \item Apply scrambling and noise to the rest.
    \item Transmit the transformed updates to the server.
\end{enumerate}
The server merges the incoming gradients using techniques like weighted averaging. Following this, a scrambling key is applied to the aggregated model before it is sent back to the clients, where it is decrypted and used for continued training.

\subsection{Security and Efficiency Trade-Off}

The proposed system achieves a balance between security and speed:
\begin{itemize}
    \item \textbf{Security:} FAS protects sensitive parameters through encryption while mitigating risks to the remaining ones through obfuscation. This layered defense shows comparable robustness to fully encrypted models in resisting attacks like model inversion and membership inference.
    \item \textbf{Efficiency:} By reducing the encryption scope and avoiding pre-processing phases like sensitivity masking, FAS cuts down overhead significantly. It is particularly well-suited for settings with limited computational power.
\end{itemize}

By integrating selective encryption, noise injection, and scrambling, this approach offers a practical and deployable solution for secure FL in privacy-sensitive areas, including medical and financial domains.

\subsection{Privacy-Preserving Techniques}

\subsubsection{Homomorphic Encryption}
Homomorphic encryption (HE) is a cryptographic scheme that permits data processing to occur while the data remains encrypted, eliminating the need for decryption during computation.

In the context of FL, HE enables secure aggregation of model updates from distributed clients. The mathematical foundation of HE ensures that:
\[
\text{Dec}(\text{Eval}(\text{Enc}(x), \text{Enc}(y))) = f(x, y),
\]
where $\text{Enc}$ and $\text{Dec}$ refer to the encryption and decryption functions, respectively, and $\text{Eval}$ denotes a computation performed on ciphertexts. This capability is essential for privacy-preserving scenarios, as it ensures that plaintext data remains inaccessible to potential adversaries. \cite{gentry2009}.

\subsubsection{Differential Privacy}
Differential Privacy (DP) is a privacy framework that adds carefully calibrated noise to statistical outputs, ensuring that individual data points have minimal impact on the overall resul. 

A mechanism $M$ is said to ensure $\epsilon$-differential privacy if, for all subsets of outputs $O$ and for any pair of datasets $D_1$ and $D_2$ that differ by only one record, the following holds:
\[
\frac{\Pr[M(D_1) \in O]}{\Pr[M(D_2) \in O]} \leq e^{\epsilon},
\]
where $\epsilon$ represents the privacy parameter.  \cite{dwork2006calibrating}.

In FL, DP adds noise to model updates or gradients prior to aggregation, shielding individual data contributions while retaining overall model performance. This makes DP particularly useful for safeguarding privacy in high-stakes sectors like medical and financial applications.

\subsection{Process Flow of the Proposed Model}
The proposed model follows a three-stage pipeline: client-side processing, server-side processing, and client post-processing. Each phase includes well-defined steps illustrated with pseudocode and supported by visual diagrams.

\begin{figure}[!t]
    \centering
    \includegraphics[width=\linewidth]{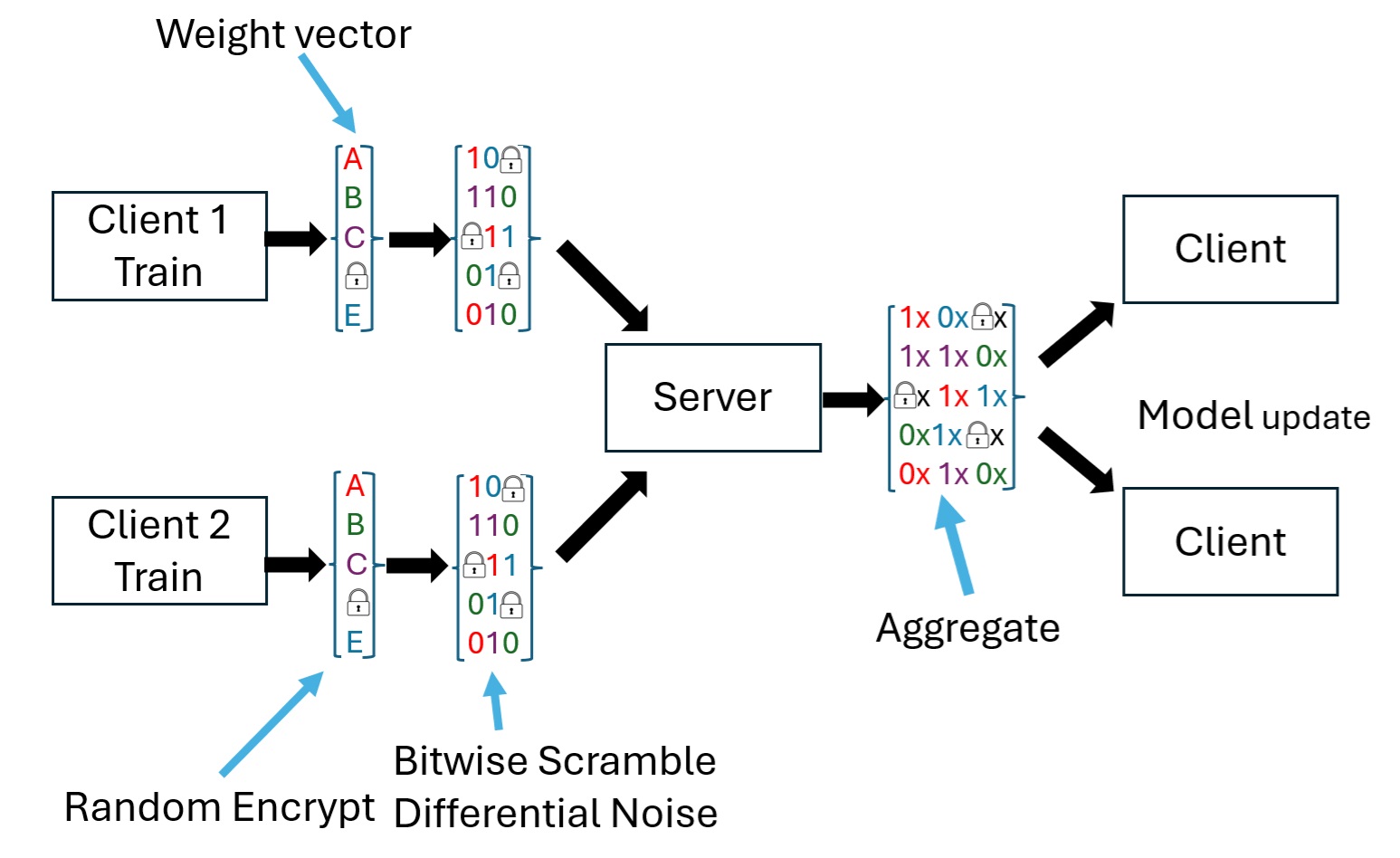}
    \caption{General depiction of the FL model showing the processes of encryption, scrambling, and aggregation.}
    \label{fig-neural-network}
\end{figure}
Figure \ref{fig-neural-network} provides a general depiction of our FL model, demonstrating the encryption, scrambling, and noise addition steps across client nodes, and the subsequent aggregation performed by the server.
 
\begin{figure}[!t] \centering \includegraphics[width=0.45\textwidth]{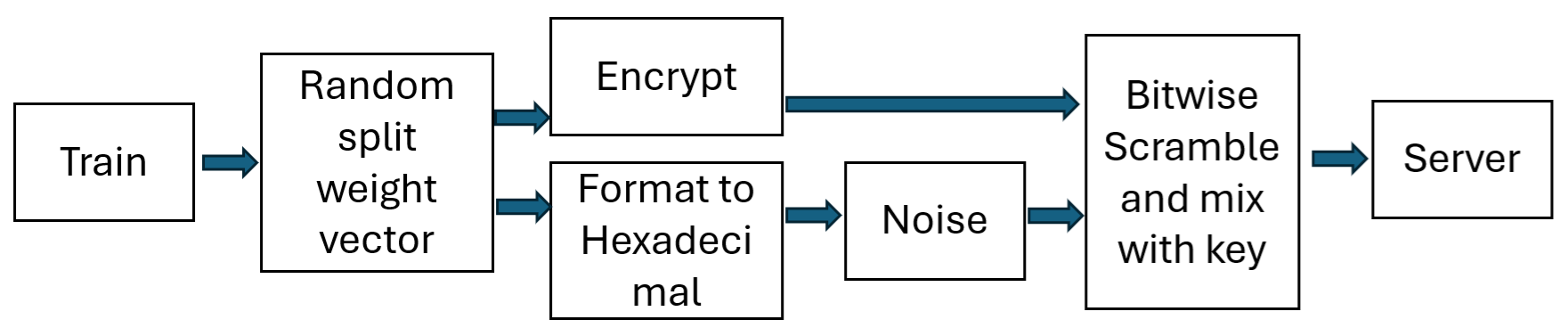} \caption{Client-side process in FL system: encryption, scrambling, and noise addition.} \label{fig-model1re3}
\end{figure}

\begin{figure}[!t] \centering \includegraphics[width=0.45\textwidth]{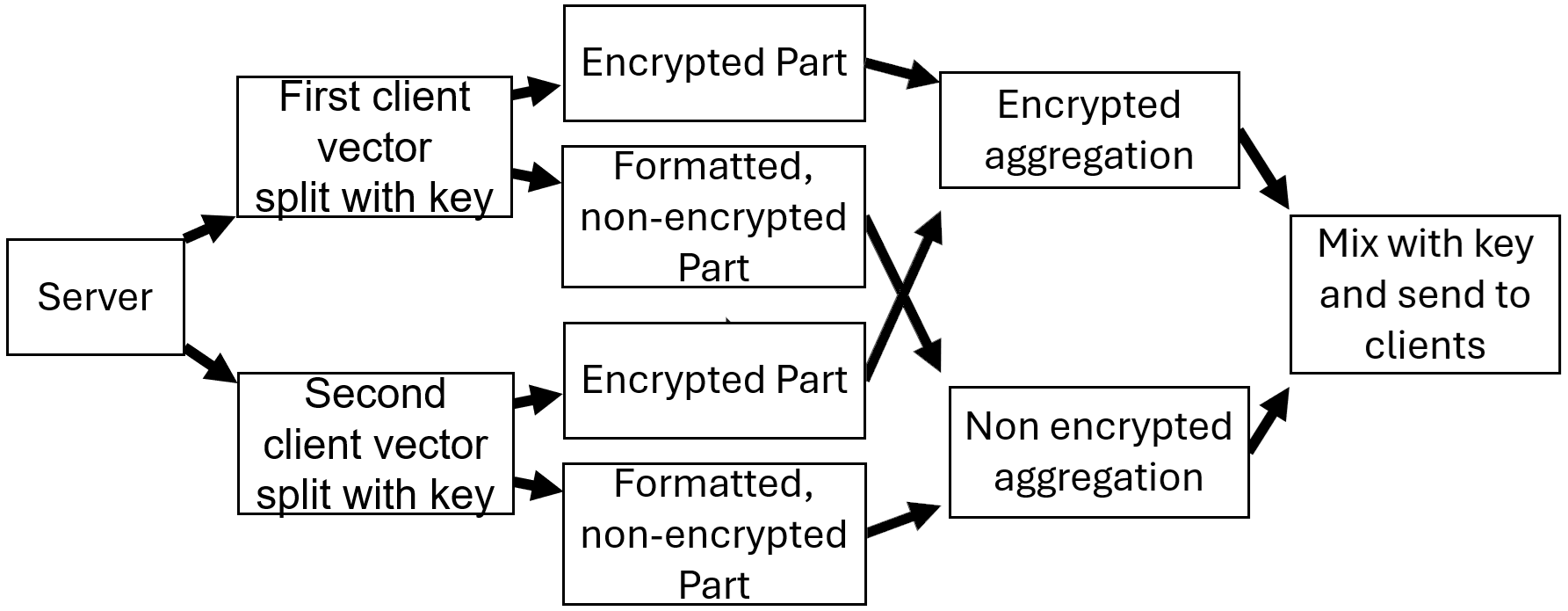} \caption{Server-side process in FL system: aggregation and re-scrambling.} \label{fig-model2re3} 
\end{figure}

Figure~\ref{fig-neural-network} presents a general overview of the FL system, encompassing client-side encryption, scrambling, noise addition, and server-side aggregation. Figures~\ref{fig-model1re3} and~\ref{fig-model2re3} illustrate the detailed client and server sides workflows, respectively.

The process begins at the client side, where local data is utilized to train a model, producing a weight vector. This weight vector is split, encrypted, and further obfuscated through bitwise scrambling and the addition of noise.

Once encrypted, the weight vectors are formatted to hexadecimal and transmitted to the server. The server subsequently collects encrypted model data from all clients and applies federated averaging (FedAvg) to combine the parameters. Importantly, the aggregation is done on both encrypted and non-encrypted vectors, ensuring privacy is maintained while also allowing for more efficient computation.

After aggregation, the server mixes the results with a scrambling key before sending the updated parameters back to the clients, where they are decrypted and applied to the local models for further training.

This privacy-aware and secure approach facilitates effective model training on decentralized data, while safeguarding sensitive information from unauthorized access.

\subsubsection{Client-Side Processing}
In the client-side processing stage, model weights ($w$) are prepared for secure transmission to the server by encrypting selected weights, obfuscating the remaining weights, and combining them into a single data structure. Initially, the weights are split into two subsets based on a predefined encryption percentage ($enc\_pct$) (Lines 6--8). The selected portion of the weights is encrypted using a HE function ($enc\_func$), and these encrypted weights are stored in $enc\_w$ (Lines 9--11). For the remaining weights, they are first formatted to resemble encrypted data ($fmt\_weight$), noise is added using $noise\_func$, and the noisy weights are scrambled with a cryptographic key ($scr\_key$) via a scrambling function ($scr\_func$) (Lines 12--16). These scrambled weights are stored in $scr\_w$. Finally, the encrypted and scrambled weights are combined into a single structure ($enc\_scr\_w$) (Line 17) and transmitted to the server (Line 18).

\begin{algorithm}
\caption{Client-Side Enc, Scrambling, and Noise Addition}
\begin{algorithmic}[1]
\Require
\State $w$: Model weights
\State $enc\_pct$: Percentage of weights to encrypt
\State $enc\_func$: Encryption function
\State $scr\_key$: Scrambling key
\State $scr\_func$: Scrambling function
\State $noise\_func$: Noise addition function
\Ensure 
\State $enc\_scr\_w$: Processed (encrypted and scrambled) weights
\Procedure{Client Enc-Scr-Noise}{}
\State Split $w$ into encrypted and non-encrypted parts
\State $n \gets$ length of $w$
\State $num\_enc \gets (enc\_pct / 100) * n$
\State $enc\_w \gets []$
\State $scr\_w \gets []$

\For{$i = 0$ to $num\_enc - 1$}
    \State $enc\_weight \gets enc\_func(w[i])$
    \State Append $enc\_weight$ to $enc\_w$
\EndFor

\For{$i = num\_enc$ to $n - 1$}
    \State $fmt\_weight \gets$ format\_as\_encrypted($w[i]$)
    \State $noisy\_weight \gets noise\_func($fmt\_weight$)$
    \State $scr\_weight \gets scr\_func($noisy\_weight$, $scr\_key$)$
    \State Append $scr\_weight$ to $scr\_w$
\EndFor

\State Combine $enc\_w$ and $scr\_w$ into $enc\_scr\_w$
\State Send $enc\_scr\_w$ to the server
\EndProcedure
\end{algorithmic}
\end{algorithm}

\subsubsection{Server-Side Processing}
On the server side, the received combined weights ($enc\_scr\_w$) are processed to aggregate updates securely without decryption. The scrambled weights are first unscrambled using the same cryptographic key ($scr\_key$) applied during client-side processing (Lines 5--7). Once unscrambled, these weights are aggregated using a processing function ($srv\_proc\_func$) (Lines 8--9). Homomorphic operations are applied directly on the encrypted weights without decryption (Line 9). After processing, the unscrambled weights are reformatted and scrambled again using $scr\_func$ and $scr\_key$ to preserve security (Lines 10--13). These processed and re-scrambled weights are then combined with the processed encrypted weights and sent back to the client (Line 14).

\begin{algorithm}

\caption{Server-Side Processing without Decryption}
\begin{algorithmic}[1]
\Require 
\State $enc\_scr\_w$: Received weights from client
\State $scr\_key$: Scrambling key for re-scrambling
\State $srv\_proc\_func$: Server processing function
\Ensure 
\State $sec\_w$: Processed weights sent back to client
\Procedure{Server Proc w/o Decrypt}{}
\State Receive $enc\_scr\_w$ from the client
\State $proc\_w \gets []$
\State $unscr\_w \gets []$

\For{$w$ in $scr\_w$}
    \State $unscr\_weight \gets unscramble\_func(w, scr\_key)$
    \State Append $unscr\_weight$ to $unscr\_w$
\EndFor

\State Process $unscr\_w$ using $srv\_proc\_func$
\State Process $enc\_w$ using homomorphic operations (no decryption needed)

\State Format $unscr\_w$ back to scrambled format:
\For{$w$ in $unscr\_w$}
    \State $re\_scr\_w \gets scr\_func(w, scr\_key)$
    \State Append $re\_scr\_w$ to $proc\_w$
\EndFor

\State Combine $proc\_enc\_w$ with $re\_scr\_w$
\State Send $sec\_w$ back to the client
\EndProcedure
\end{algorithmic}
\end{algorithm}

\subsubsection{Client Post-Processing}
In the final stage, the client receives the processed weights ($sec\_w$) from the server (Line 2) and restores them for updating the local model. The encrypted subset of weights is decrypted using the appropriate decryption function (Lines 3--4), while the scrambled weights are unscrambled with the cryptographic key ($scr\_key$) to restore their original form (Lines 5--7). These two subsets of weights are then combined to reconstruct the full model weights (Line 8), which are subsequently used to refine the model instance residing on the client for the next round of federated learning. (Line 9).

\begin{algorithm}
\vspace{-1mm}
\caption{Client Post-Processing with Noise Restoration}
\begin{algorithmic}[1]
\Require 
\State $sec\_w$: Weights received from server
\State $scr\_key$: Scrambling key for restoration
\Ensure 
\State $final\_w$: Updated model weights for next round
\Procedure{Client Post-Proc with Noise Rest}{}
\State Receive $sec\_w$ from the server
\State $proc\_enc\_w \gets sec\_w[0 : num\_enc]$
\State $proc\_scr\_w \gets sec\_w[num\_enc : ]$

\For{$w$ in $proc\_scr\_w$}
    \State $restored\_w \gets reverse\_scr\_func(w, scr\_key)$
    \State Append $restored\_w$ to $rest\_w$
\EndFor

\State Combine $proc\_enc\_w$ and $rest\_w$ to form $final\_w$
\State Proceed to next round of FL with $final\_w$
\EndProcedure
\end{algorithmic}
\end{algorithm}

\subsection{Advantages of the Proposed Model}

\begin{itemize}
    \item \textbf{Enhanced Security:} The combination of encryption, noise, and scrambling protects both encrypted and unencrypted data.
    \item \textbf{Efficiency:} Reduces computational overhead by selectively encrypting only a portion of the weights.
    \item \textbf{Scalability:} Optimized for large-scale FL systems with multiple clients.
\end{itemize}

\subsection{Privacy Analysis}

This section presents a provide a formal analysis demonstrating that our FL mechanism—integrating selective homomorphic encryption, differential privacy, and bitwise scrambling—upholds a clear and measurable differential privacy guarantee.

\subsubsection{Setup and Definitions}

We consider a FL setting with $n$ clients, each holding private data. The global model is described by parameters indexed by $[N]$. Let $S \subseteq [N]$ be the subset of parameters to be protected by selective homomorphic encryption , and let $[N]\setminus S$ be the remaining parameters to which we add differential privacy (DP) noise.

\paragraph{Differential Privacy.}
A mechanism $M$ satisfies $\epsilon$-differential privacy if, for any neighboring datasets $D$ and $D'$ (differing in one element), and for every measurable set of outputs $O$, it holds that:

\[
\frac{\Pr[M(D)\in O]}{\Pr[M(D')\in O]} \le e^\epsilon.
\]\cite{dwork2006calibrating}

\paragraph{Selective Homomorphic Encryption.}
Parameters in $S$ are encrypted using a semantically secure HE scheme. Semantic security ensures no polynomial-time adversary can distinguish ciphertexts of different messages, implying no additional privacy cost in a DP sense (0-DP).\cite{gentry2009}

\paragraph{Differential Privacy on the Remaining Parameters.}
For each $i \in [N]\setminus S$, we add noise calibrated to $\epsilon_i$-DP. By the composition property of DP, releasing all these parameters together is $\left(\sum_{i \in [N]\setminus S}\epsilon_i\right)$-DP.\cite{dwork2006calibrating}

\paragraph{Bitwise Scrambling.}
We define a scrambling function $\mathcal{T}_k:\mathcal{X} \to \mathcal{X}$, where $\mathcal{X}$ is the space of possible model outputs. $\mathcal{T}_k$ is deterministic, keyed by $k$, and independent of $D$ except through $M(D)$. This scrambling is a form of post-processing.\cite{Dwork2006}

\subsubsection{Privacy Guarantee}

\begin{theorem}

Consider a mechanism $\mathcal{M}$ that on input dataset $D$:
\begin{enumerate}
    \item Encrypts $W_i$ for $i \in S$ using a semantically secure HE scheme \cite{gentry2009,Fontaine2007}.
    \item Adds noise to each $W_j$ for $j \in [N]\setminus S$ to achieve $\epsilon_j$-DP. Releasing all noisy parameters together is $\left(\sum_{j \in [N]\setminus S}\epsilon_j\right)$-DP \cite{Dwork2014}. 
    \item Applies a scrambling function $\mathcal{T}_k$ to the entire output, resulting in $\mathcal{M}'(D)=\mathcal{T}_k(\mathcal{M}(D))$ \cite{Dwork2006}.
\end{enumerate}
Then, $\mathcal{M}'(D)$ is also $\left(\sum_{j \in [N]\setminus S}\epsilon_j\right)$-DP.
\end{theorem}

\begin{proof}
\noindent \textbf{Step 1 (DP on $[N]\setminus S$):}
For $[N]\setminus S$, each parameter $W_j + \text{noise}_j$ is $\epsilon_j$-DP. By composition, releasing all these parameters is $\sum_{j \in [N]\setminus S}\epsilon_j$-DP \cite{Dwork2014}. Formally, for any adjacent $D,D'$ and event $O$:
\[
\frac{\Pr[\mathcal{M}_{[N]\setminus S}(D)\in O]}{\Pr[\mathcal{M}_{[N]\setminus S}(D')\in O]} \le e^{\sum_{j \in [N]\setminus S}\epsilon_j\cite{Dwork2014}}.
\].

\noindent \textbf{Step 2 (Selective Encryption on $S$):}
Parameters in $S$ are encrypted with a semantically secure HE scheme. Without the secret key, no information about the plaintext is revealed, contributing 0-DP cost \cite{gentry2009,Fontaine2007}.

\noindent \textbf{Step 3 (Scrambling as Post-Processing):}
The scrambling function $\mathcal{T}_k$ operates as a deterministic transformation that does not depend on the underlying data. A core principle of differential privacy is that its guarantees are preserved under any data-independent transformation. Specifically, if a mechanism $\mathcal{M}$ satisfies $\epsilon$-DP, then so does $f(\mathcal{M})$ for any deterministic function $f$~\cite{Dwork2006,Dwork2014}.

Combining these results:
\[
\frac{\Pr[\mathcal{M}'(D)\in O]}{\Pr[\mathcal{M}'(D')\in O]} = \frac{\Pr[\mathcal{T}_k(\mathcal{M}(D))\in O]}{\Pr[\mathcal{T}_k(\mathcal{M}(D'))\in O]} \le e^{\sum_{j \in [N]\setminus S}\epsilon_j}.
\]

Thus, $\mathcal{M}'(D)$ maintains the same DP guarantee as $\mathcal{M}(D)$.
\end{proof}

\section{Evaluation}

This section evaluates our FL setup using various security techniques, including HE, differential privacy, and our proposed FAS. We assess their impact on test accuracy, communication cost, computational overhead, and privacy preservation.

Each dataset was divided into 10 equal parts and assigned to separate machines operating as Flower~\cite{flower} clients. Each client trained its local model over 10 rounds, performing 20 epochs per round. A validation set, comprising 10\% of the training data, was used to maintain consistency, while final accuracy was measured on an independent test set. All models were trained using CPU resources to simulate realistic constraints.

\subsection{Datasets Used in Experiments}

We selected publicly available datasets covering diverse imaging contexts with high usability ratings. Our criteria included broad disease coverage, diverse imaging techniques, and sufficient training and validation data.

\subsubsection{CIFAR-10 Dataset~\cite{CIFAR10}} A benchmark for image classification, CIFAR-10 consists of 60,000 32x32 images across 10 classes. It includes 50,000 training and 10,000 test images, commonly used for model performance evaluation.

\subsubsection{Chest X-ray (COVID-19, Pneumonia) Dataset~\cite{ChestXray}} This dataset contains 6,339 grayscale X-ray images across three classes: benign, COVID-19, and pneumonia (2,313 images per class). These images aid in lung condition assessment.

\begin{figure}[!t]
    \centering
    \begin{tabular}{ccc}
    \includegraphics[scale=0.35]{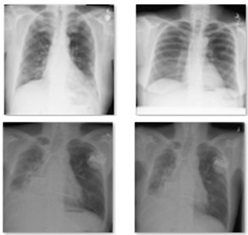} &
    \includegraphics[scale=0.35]{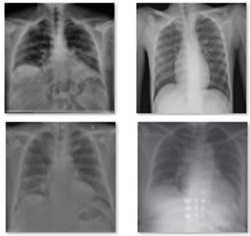} &
    \includegraphics[scale=0.35]{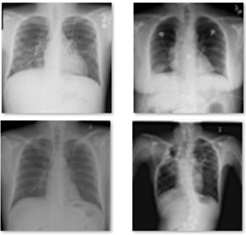} \\
    (a) Benign & (b) COVID-19 & (c) Pneumonia
    \end{tabular}
    \caption{Examples from the Chest X-ray (Covid, Pneumonia) dataset~\cite{mine}}
    \label{fig-dataset-Xray}
\end{figure}

\subsubsection{CT Kidney Dataset~\cite{kidney}} Comprising 12,446 CT scans categorized as benign (5,077), cyst (3,709), stone (1,377), and tumor (2,283), this dataset aids in kidney disease diagnosis.

\begin{figure}[!t]
    \centering
    \begin{tabular}{cc}
    \includegraphics[scale=0.6]{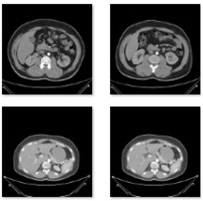} &
    \includegraphics[scale=0.6]{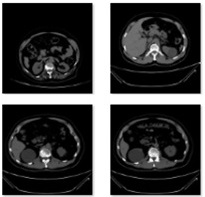} \\
    (a) Benign & (b) Cyst \\
    \includegraphics[scale=0.6]{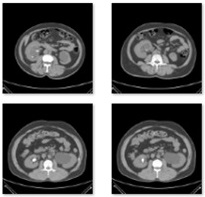} &
    \includegraphics[scale=0.6]{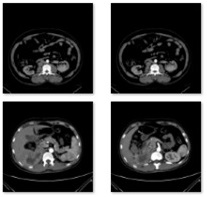} \\
    (c) Stone & (d) Tumor
    \end{tabular}
    \caption{Examples from the CT Kidney dataset~\cite{mine}}
    \label{fig-dataset-kidney}
\end{figure}

\subsubsection{Diabetic Retinopathy Dataset~\cite{diabetic}} This dataset comprises 88,645 fundus images, categorized as benign (65,342) or diabetic retinopathy (23,303). Fundus imaging assists in detecting optic nerve abnormalities.

\begin{figure}[!t]
    \centering
    \begin{tabular}{cc}
    \includegraphics[scale=0.4]{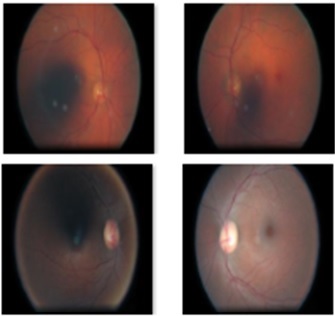} &
    \includegraphics[scale=0.4]{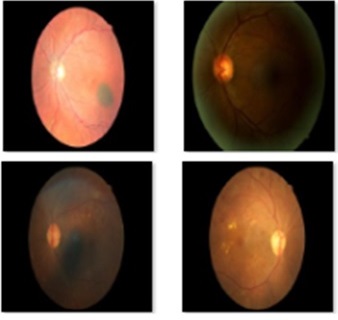} \\
    (a) Benign & (b) Diabetic retinopathy
    \end{tabular}
    \caption{Examples from the Diabetic Retinopathy Detection dataset~\cite{mine}}
    \label{fig-dataset-diabetic}
\end{figure}

\subsubsection{Lung Cancer Histopathological Images~\cite{lungcolon}} This dataset contains 15,000 histopathological images, equally distributed across benign, squamous cell carcinoma (SCC), and adenocarcinoma (ACC) classes.

\begin{figure}[!t]
    \centering
    \begin{tabular}{ccc}
    \includegraphics[scale=0.35]{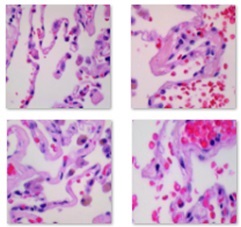} &
    \includegraphics[scale=0.35]{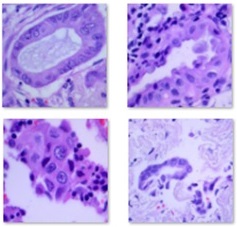} &
    \includegraphics[scale=0.35]{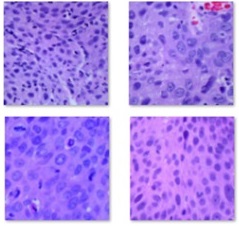} \\
    (a) Benign & (b) ACC & (c) SCC
    \end{tabular}
    \caption{Examples from the Lung Cancer Histopathological Images dataset.~\cite{mine}}
    \label{fig-dataset-lung}
\end{figure}

\section{Experimental Setup}

This section outlines the experimental framework employed to assess different cryptographic methods within a FL environment applied to medical imaging datasets.

\subsection{CloudLab Environment}
Experiments were conducted on CloudLab~\cite{cloudlab}, a cloud computing testbed for systems research. The experimental setup consisted of 11 standalone physical machines in a shared-nothing cluster, interconnected through 10 Gbps Ethernet. Each machine featured dual Intel E5-2683 v3 CPUs (14 cores at 2.00 GHz), 256 GB of RAM, and a pair of 1 TB hard drives, offering ample computational capacity for FE workloads.

\subsection{Federated Learning Framework: Flower}
Flower~\cite{flower}, an open-source FL framework, was used for model training and evaluation, supporting TensorFlow~\cite{tensorflow}, PyTorch~\cite{pytorch}, and MXNet~\cite{mxnet}. One CloudLab machine served as the Flower server, while the remaining 10 functioned as Flower clients.

\subsection{Federated Learning Setup}
Each dataset was divided into 10 equal partitions, distributed across the Flower clients. The centralized FL setup comprised 10 clients and a server. Each model was trained over 10 rounds, with 20 epochs per round. A validation set (10\% of training data) monitored performance, and final accuracy was evaluated on a test set. All training used CPUs to simulate resource-constrained environments.

\subsection{Rationale for Centralized Federated Learning}
Decentralized FL was not used due to a lack of mutual trust among clients. A centralized setup ensures data privacy, as clients communicate only with the server, avoiding direct data exchange and aligning with confidentiality requirements.

\subsection{Deep Learning Models}
We evaluated four deep learning models:

\subsubsection{ResNet-50\cite{resnet}}
ResNet-50 mitigates the vanishing gradient problem with skip connections, facilitating deep network training. 

\subsubsection{DenseNet121\cite{densenet}}
DenseNet121 enhances feature reuse by connecting each layer to all previous layers, improving efficiency in resource-constrained settings.

\subsubsection{EfficientNetB0\cite{efficientnet}}
EfficientNetB0 optimally scales network width, depth, and resolution, balancing accuracy and efficiency.

\subsubsection{MobileNet V2\cite{mobilenet}}
MobileNet V2 leverages efficient convolutional operations to minimize computational overhead while still delivering strong predictive performance, making it well-suited for deployment on mobile and embedded devices.

\subsection{Encryption Configurations}
We compared three encryption configurations:
\begin{itemize}
    \item \textbf{Non-Enc}: No encryption applied.
    \item \textbf{Full-Enc}: All model weights encrypted using HE.
    \item \textbf{Partly-Enc}:FAS is applied to a subset of data, with scrambling for security.
\end{itemize}

\textbf{Full Overhead} represents additional training time due to full encryption, while \textbf{Partly Overhead} captures the impact of FAS.

\begin{table}[!t]
\centering
\scriptsize
\caption{Training Time (Minutes) and Overheads Across Datasets}
\begin{tabular}{|p{1.7cm}|p{2.5cm}|p{1.4cm}|p{1.5cm}|}
\hline
\textbf{Model} & \textbf{Dataset} & \textbf{Full Enc (min)} & \textbf{Partly Enc (min)} \\ \hline
\multirow{5}{*}{EffNetB0}  
  & CIFAR-10    & 111  & 21   \\ \cline{2-4}
  & CT Kidney   & 228  & 127  \\ \cline{2-4}
  & Lung        & 355  & 254  \\ \cline{2-4}
  & COVID       & 698  & 591  \\ \cline{2-4}
  & Diabetic\newline Retinopathy & 1298 & 1351 \\ \hline
\multirow{5}{*}{DenseNet121}  
  & CIFAR-10    & 219  & 35   \\ \cline{2-4}
  & CT Kidney   & 514  & 326  \\ \cline{2-4}
  & Lung        & 506  & 312  \\ \cline{2-4}
  & COVID       & 1542 & 1351 \\ \cline{2-4}
  & Diabetic\newline Retinopathy & 1542 & 1351 \\ \hline
\multirow{5}{*}{MobileNetV2}  
  & CIFAR-10    & 68   & 15   \\ \cline{2-4}
  & CT Kidney   & 163  & 100  \\ \cline{2-4}
  & Lung        & 467  & 231  \\ \cline{2-4}
  & COVID       & 610  & 555  \\ \cline{2-4}
  & Diabetic\newline Retinopathy & 610  & 555  \\ \hline
\multirow{5}{*}{ResNet-50}  
  & CIFAR-10    & 530  & 68   \\ \cline{2-4}
  & CT Kidney   & 861  & 399  \\ \cline{2-4}
  & Lung        & 834  & 372  \\ \cline{2-4}
  & COVID       & 1140 & 673  \\ \cline{2-4}
  & Diabetic\newline Retinopathy & 1140 & 673  \\ \hline
\end{tabular}
\label{Table:Training Time and Overhead Across Datasets}
\end{table}

Table~\ref{Table:Training Time and Overhead Across Datasets} presents training times (minutes) across datasets. Fully encrypted models require significantly more training time, with ResNet-50 showing the highest overhead. In contrast, FAS reduce overhead substantially while maintaining security.

\begin{figure}[!t]
    \centering
    \includegraphics[width=\linewidth]{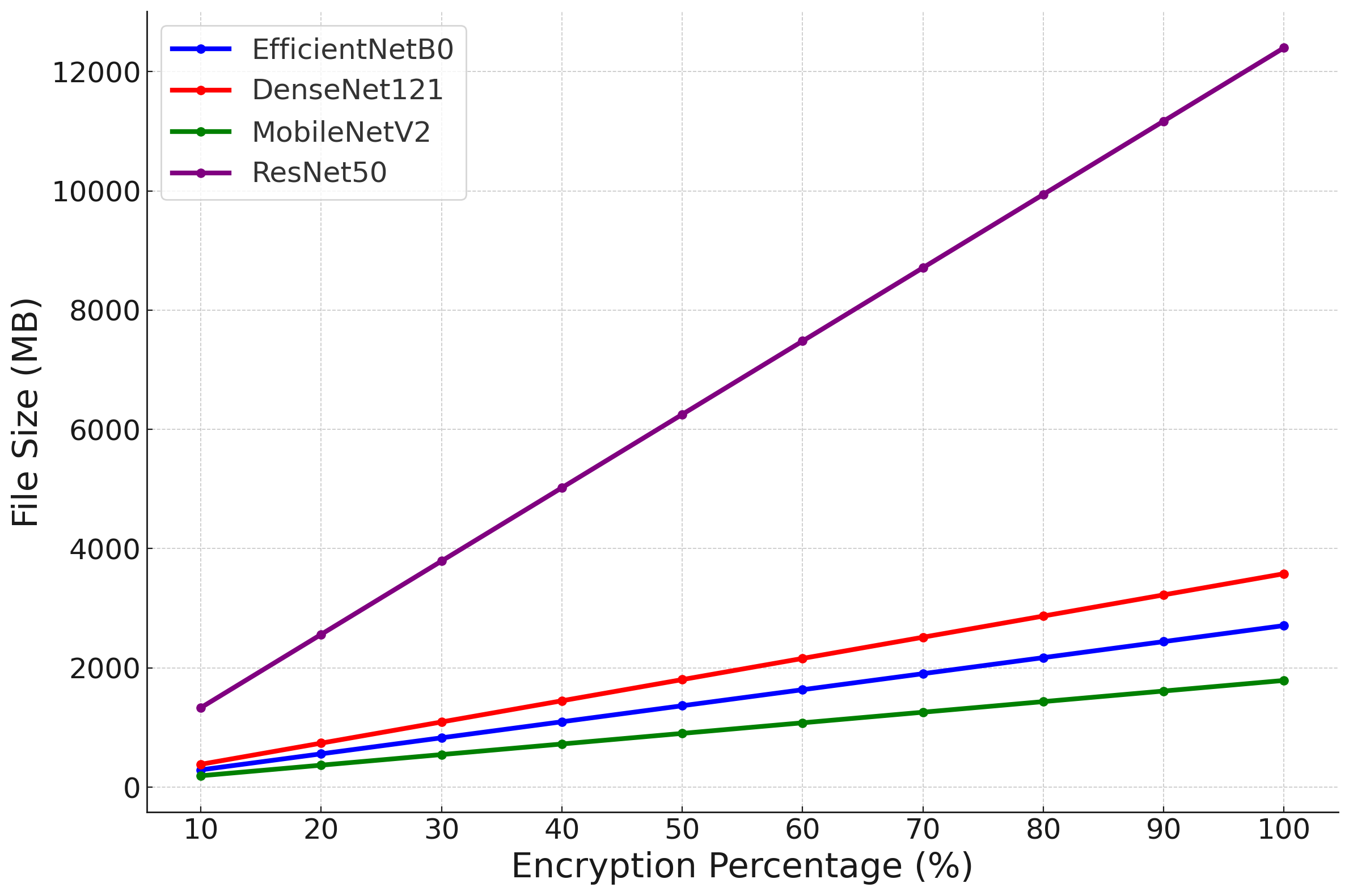}
    \caption{Encrypted file size per encryption level (10\% - 100\%)}
    \label{fig:size_graph}
\end{figure}

Figure~\ref{fig:size_graph} shows communication costs across encryption levels, revealing higher costs for fully encrypted models, whereas FAS optimizes performance.

\begin{figure}[!t]
    \centering
    \includegraphics[width=\linewidth]{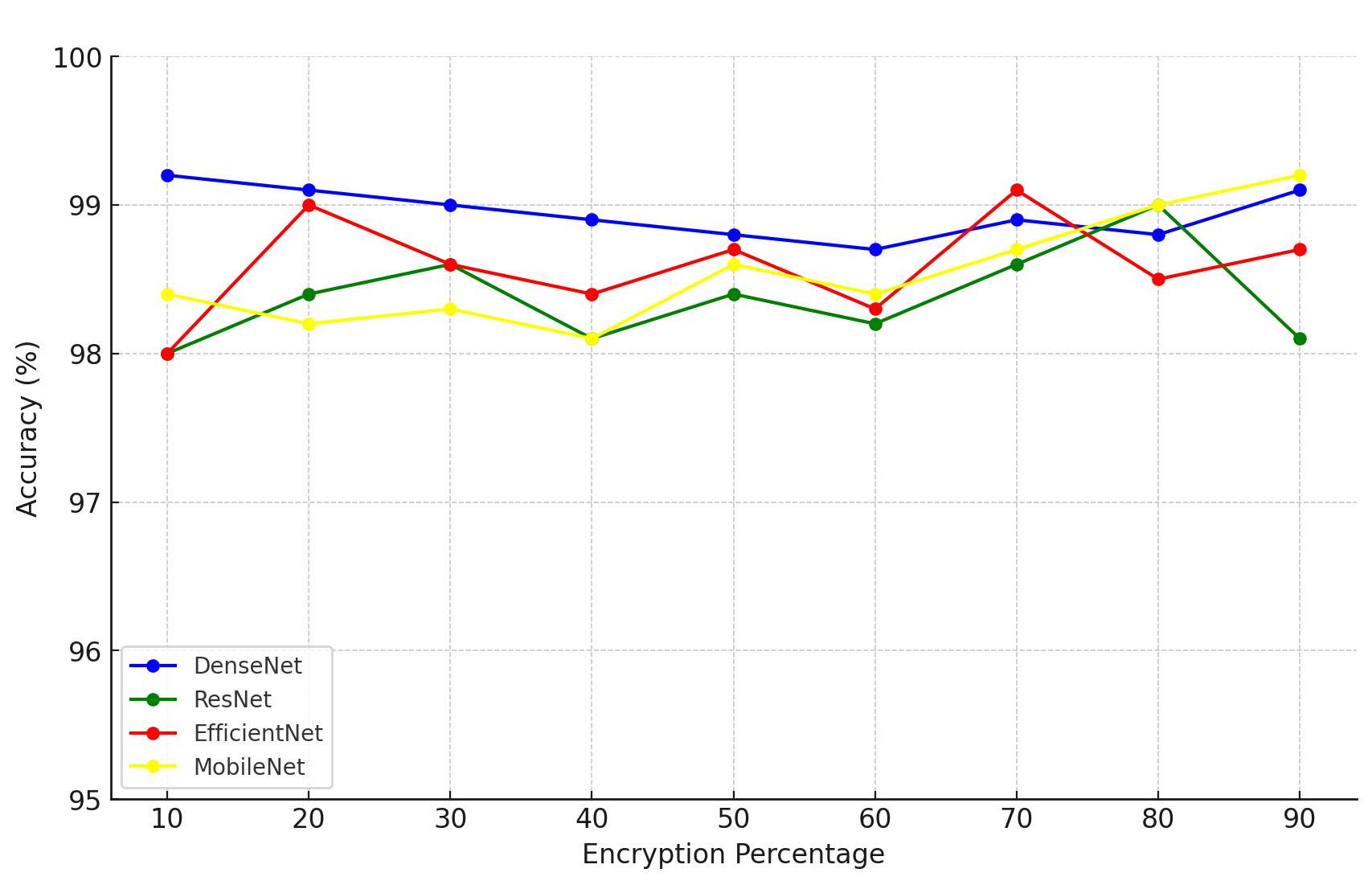} 
    \caption{Accuracy Vs Encryption Percentage Across Different Models}
    \label{fig:accuracy_comparison}
\end{figure}

Figure~\ref{fig:accuracy_comparison} indicates that accuracy remains stable across encryption percentages, confirming the effectiveness of FAS in preserving model performance.

\begin{figure}[!t]
    \centering
    \includegraphics[width=\linewidth]{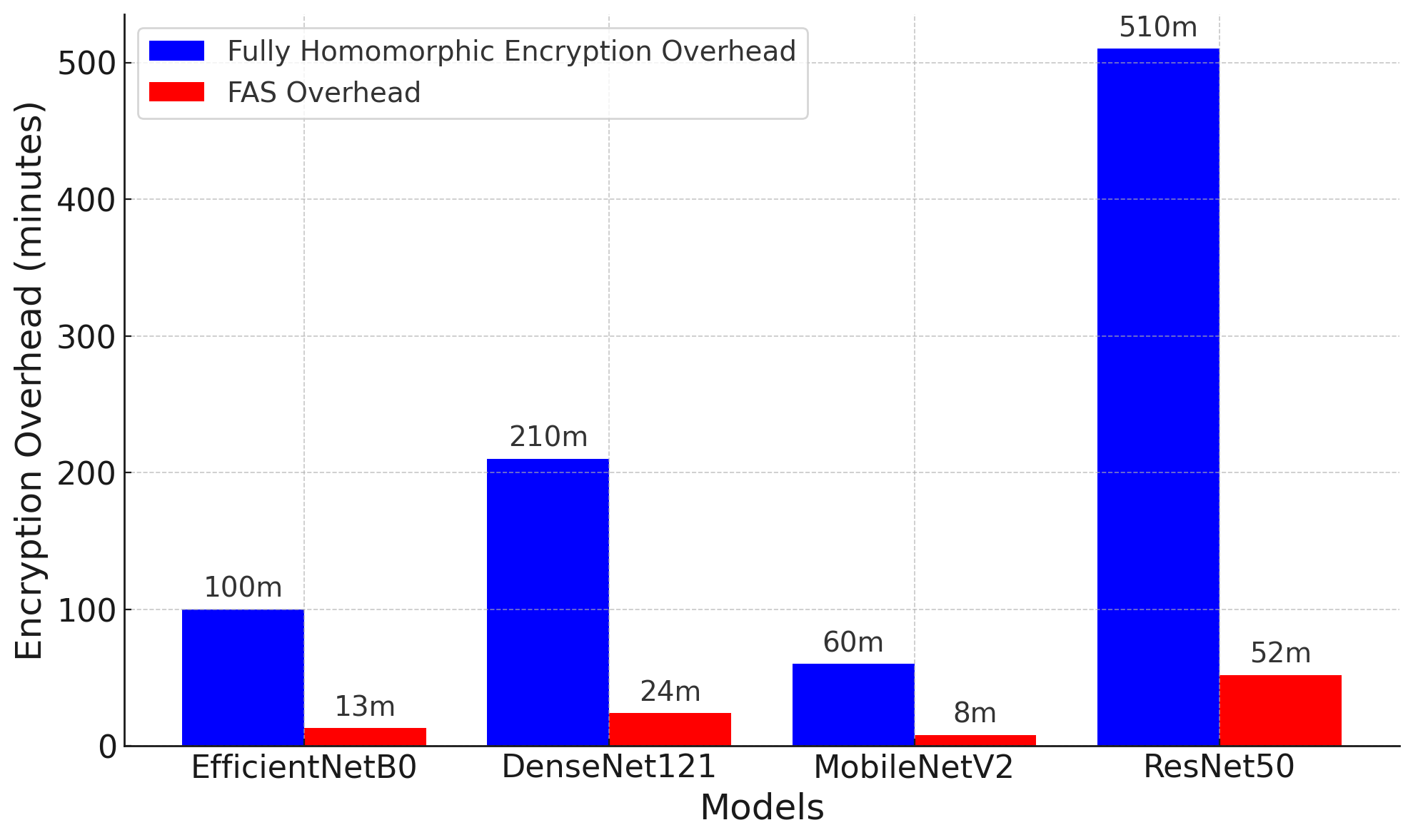} 
    \caption{Training Time and Encrypted Size Comparison for FHE and FAS Models (CT Kidney dataset)}
    \label{fig:encryption_comparison}
\end{figure}

Figure~\ref{fig:encryption_comparison} highlights significant reductions in encryption overhead when using FAS, making FAS a viable option for resource-limited FL scenarios.

\subsection{Effects on Security}

Before conducting our main security experiments, we implemented several security tests to evaluate the robustness of our encryption methods. These tests included integrity checks, timing attack tests, noise and precision evaluations, and simulations of chosen-ciphertext and chosen-plaintext attacks. Additional differential attack tests and statistical distribution checks were performed to verify that encryption preserved data integrity and privacy. Our findings indicate that 10\% FAS encryption successfully passed the same security tests as 100\% FHE, demonstrating strong security with reduced computational overhead.

\begin{table}[!t]
\scriptsize
\caption{FHE vs FAS: Security Test Results}
\centering
\begin{tabular}{|l|c|c|}
\hline
\textbf{Test} & \textbf{FAS} & \textbf{FHE} \\ \hline
Integrity Check              & False & False \\ \hline
Chosen-Plaintext Attack       & True  & True  \\ \hline
Chosen-Ciphertext Attack      & True  & True  \\ \hline
Noise and Precision Test     & True  & True  \\ \hline
Timing Attack Test           & True  & True  \\ \hline
Differential Attack Test     & True  & True  \\ \hline
Statistical Dist. Test       & True  & True  \\ \hline
Homomorphism Test            & True  & True  \\ \hline
Leakage Test                 & True  & True  \\ \hline
\end{tabular}
\label{Table:FHE vs FAS: Security Test Results}
\end{table}

Table~\ref{Table:FHE vs FAS: Security Test Results} compares FAS with FHE across a series of security evaluations. The results confirm that FAS provides comparable protection, reinforcing its viability as a lightweight yet effective security mechanism.

\subsubsection{Evaluating Privacy with MSSIM and VIFP}

Model inversion attacks pose a threat to FL by reconstructing training data from model updates. To assess privacy resilience, we use Mean Structural Similarity Index (MSSIM) and Visual Information Fidelity in the Pixel domain (VIFP).

MSSIM quantifies structural similarity between original and reconstructed images, with lower scores indicating stronger privacy. VIFP measures visual quality based on perceptual models, highlighting the impact of noise and scrambling techniques. Both metrics are widely recognized for privacy evaluation in adversarial scenarios.

Our results show that selective encryption with scrambling and noise significantly degrades reconstruction quality. MSSIM stabilizes at 52\% for 20\% encryption, and VIFP follows a similar trend across multiple encryption levels. This confirms that our approach effectively mitigates inversion attacks while optimizing computational efficiency.

Further tests demonstrated that even at 10\% encryption, reconstructed images showed significant distortion. The MSSIM score at 10\% encryption was 58\%, highlighting a substantial deviation from the original images. As encryption increased, the score stabilized at around 52\%, reinforcing the effectiveness of partial encryption in maintaining privacy.

\begin{figure}[!t]
    \centering
    \includegraphics[width=\linewidth]{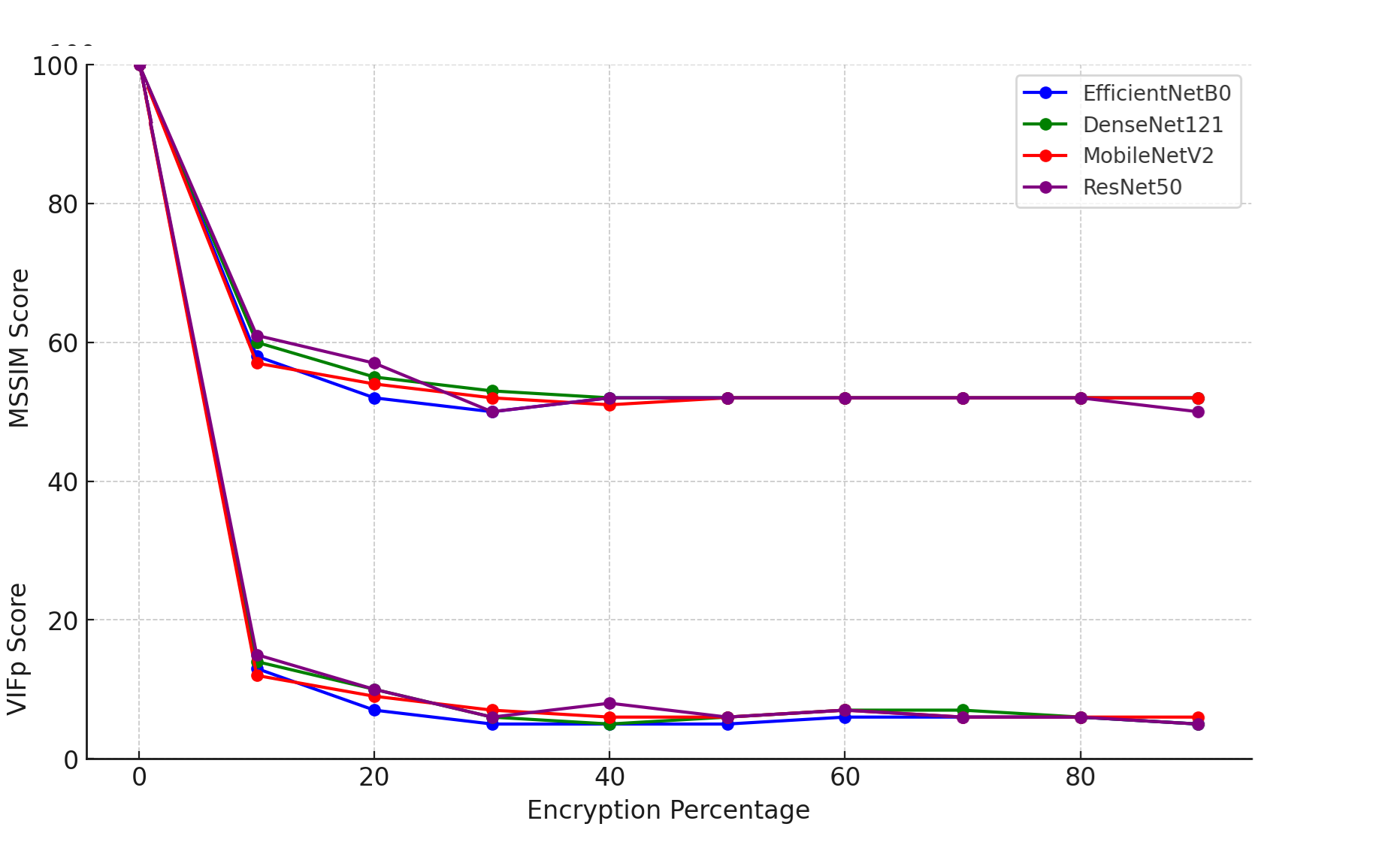}
    \caption{MSSIM and VIFP Scores across Encryption Percentages for Various Models}
    \label{fig:mssim_graph}
\end{figure}

Figure~\ref{fig:mssim_graph} illustrates MSSIM and VIFP scores across encryption percentages for various models. Results were averaged across multiple runs to account for variations introduced by selective encryption, scrambling, and noise. The sharp decline in scores at 10\% encryption, with stabilization around 20\%, suggests that lower encryption levels can provide security equivalent to full encryption while significantly reducing computational costs. Notably, resistance to inversion attacks at 100\% encryption is nearly identical to that at 10\% and 20\%, highlighting the robustness of selective encryption.

\subsection{Cross-testing with Related Work}

\subsubsection{Comparison with MASKCRYPT}

MASKCRYPT \cite{10506637} applies a gradient-guided encryption mask to selectively encrypt critical model weights, reducing communication overhead by up to 4.15× compared to full model encryption. However, its gradient mask is not always optimized for selecting the most sensitive gradients, leading to lower initial security scores. MASKCRYPT requires multiple rounds to stabilize encryption, as reflected in VIFP scores, which only stabilize by round 5. This indicates an adaptation period where security is suboptimal.

In contrast, FAS stabilizes earlier by combining selective homomorphic encryption, differential noise, and bitwise scrambling, encrypting a fixed percentage of weights without gradient sensitivity analysis or pre-training. This approach reduces processing time by 90\% compared to FHE while maintaining strong resistance to model inversion attacks. Unlike MASKCRYPT, which recalculates sensitivity masks each round, FAS eliminates extra computational costs and scales effectively across federated networks.

\subsubsection{Comparison with FedML-HE}

FedML-HE \cite{jin2024fedmlheefficienthomomorphicencryptionbasedprivacypreserving} identifies and encrypts critical weights based on gradient sensitivity through a pre-training phase. Our recreated FedML-HE model provides insight into its behavior, though results may differ from the original implementation. While FedML-HE offers strong privacy protection, its pre-training phase introduces significant time and resource costs. Each client generates an encryption mask independently, which, while avoiding direct mask sharing, can lead to inconsistencies in aggregated model accuracy.

Our FAS method eliminates the need for pre-training and mask aggregation by encrypting a fixed percentage of weights with differential noise and bitwise scrambling. At 10\% encryption, FAS achieves a MSSIM score of 58\%, comparable to FedML-HE’s 52\%, but with significantly lower computational costs. The lack of pre-training and per-client mask generation allows FAS to maintain efficiency and scalability in FL environments.

\begin{figure}[!t]
    \centering
    \includegraphics[width=\linewidth]{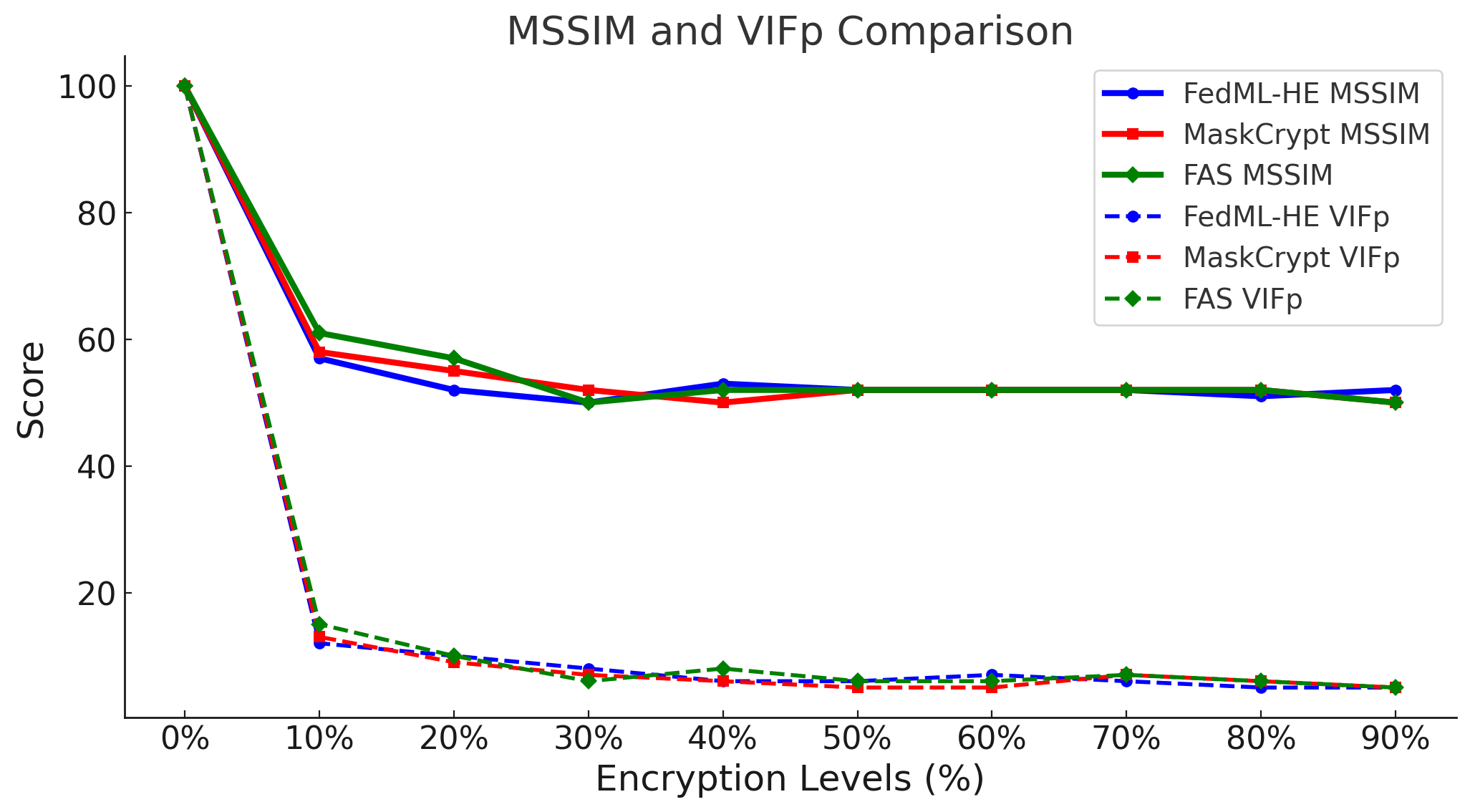}
    \caption{Comparison of MSSIM and VIFP Scores: FAS vs. FedML-HE vs. MASKCRYPT}
    \label{fig_mssim_vifp_comparison}
\end{figure}

Figure~\ref{fig_mssim_vifp_comparison} compares FAS, FedML-HE, and MASKCRYPT across MSSIM and VIFP scores. At 10\% encryption, FedML-HE shows slightly higher stability (MSSIM 52\%), while FAS achieves 58\%, and MASKCRYPT falls in between at 55\%. As encryption increases to 20\%, all methods converge with nearly identical security performance. Despite slight advantages for FedML-HE in lower encryption levels, FAS provides comparable security at significantly lower computational costs. MASKCRYPT’s need for sensitivity mask recalculation each round makes it less efficient for real-time scenarios. These findings position FAS as an optimal choice for scalable, privacy-sensitive FL applications requiring low-latency performance.

\begin{table}[!t]
\centering
\caption{Comparison of Encryption Techniques for Different Models}
\begin{tabular}{|p{2cm}|p{1.8cm}|p{0.5cm}|p{0.7cm}|p{1.1cm}|}
\toprule
   Method &       Model &  Non Enc &  Total (m) &   Overhead (m) \\
\midrule
      FAS &    EffNetB0 &       56 &          69 &               13 \\
 FEDML-HE &    EffNetB0 &       56 &          99 &               43 \\
MASKCRYPT &    EffNetB0 &       56 &          89 &               33 \\
      FAS & DenseNet121 &      152 &         176 &               24 \\
 FEDML-HE & DenseNet121 &      152 &         210 &               58 \\
MASKCRYPT & DenseNet121 &      152 &         200 &               48 \\
      FAS & MobileNetV2 &       46 &          52 &                6 \\
 FEDML-HE & MobileNetV2 &       46 &          69 &               23 \\
MASKCRYPT & MobileNetV2 &       46 &          64 &               18 \\
      FAS &   ResNet-50 &      176 &         224 &               48 \\
 FEDML-HE &   ResNet-50 &      176 &         278 &              102 \\
MASKCRYPT &   ResNet-50 &      176 &         246 &               70 \\
\bottomrule
\end{tabular}
\label{tab:comparison_table}
\end{table}

Table~\ref{tab:comparison_table} compares non-encrypted, partly encrypted, and overhead times for different models using FAS, FedML-HE, and MASKCRYPT on the Kidney dataset. The values are reconstructed based on methodologies described in the respective papers. FAS consistently has the lowest overhead across models, making it the most efficient. FedML-HE has the highest overhead due to pre-training, while MASKCRYPT offers a balance but initially has lower security, improving after multiple rounds. Overall, FAS is the best option for efficiency and scalability in FL.

\subsection{Comparative Analysis of Encryption Techniques Across Datasets}

This section compares the performance of FAS, FEDML-HE, and MASKCRYPT across five datasets: CIFAR-10, Diabetic Retinopathy, COVID, Lung, and Kidney. The analysis highlights FAS’s efficiency in minimizing overhead and encryption time compared to the other methods.

Across all datasets, FAS consistently achieves the lowest training times and computational overhead as we can see from Figures~\ref{fig_cifar_comparison}, ~\ref{fig_dia_comparison}, ~\ref{fig_kidney_comparison}, ~\ref{fig_lung_comparison}, ~\ref{fig_covid_comparison}. Unlike FEDML-HE, which incurs significant pretraining costs for sensitivity mask creation, and MASKCRYPT, which generates masks at each round, FAS employs a random masking strategy with minimal computation. This advantage is particularly evident in lightweight models like MobileNetV2 on the COVID dataset and larger models like ResNet-50 on the Lung dataset. The efficiency of FAS is further demonstrated in the Diabetic Retinopathy and Kidney datasets, where it eliminates pretraining and per-round computations, making it the most practical and scalable approach. Overall, FAS outperforms the other techniques across all datasets, confirming its robustness and adaptability.

\begin{figure}[!t]  
\centering
\includegraphics[width=\linewidth]{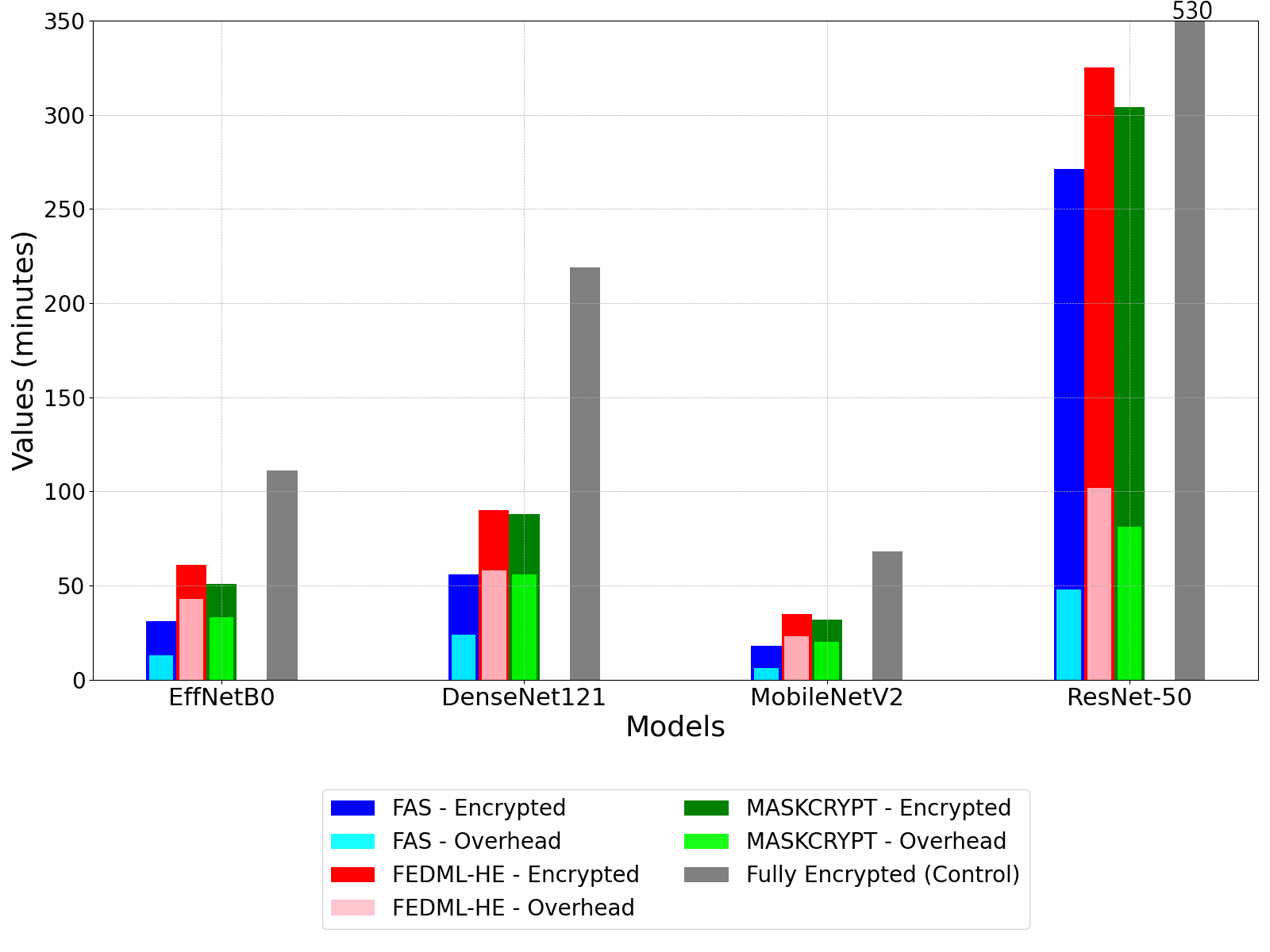}
\caption{Comparison of Partly Encrypted and Fully Encrypted Metrics Across Models (CIFAR-10).}
\label{fig_cifar_comparison}
\end{figure}

\begin{figure}[!t]
\centering
\includegraphics[width=\linewidth]{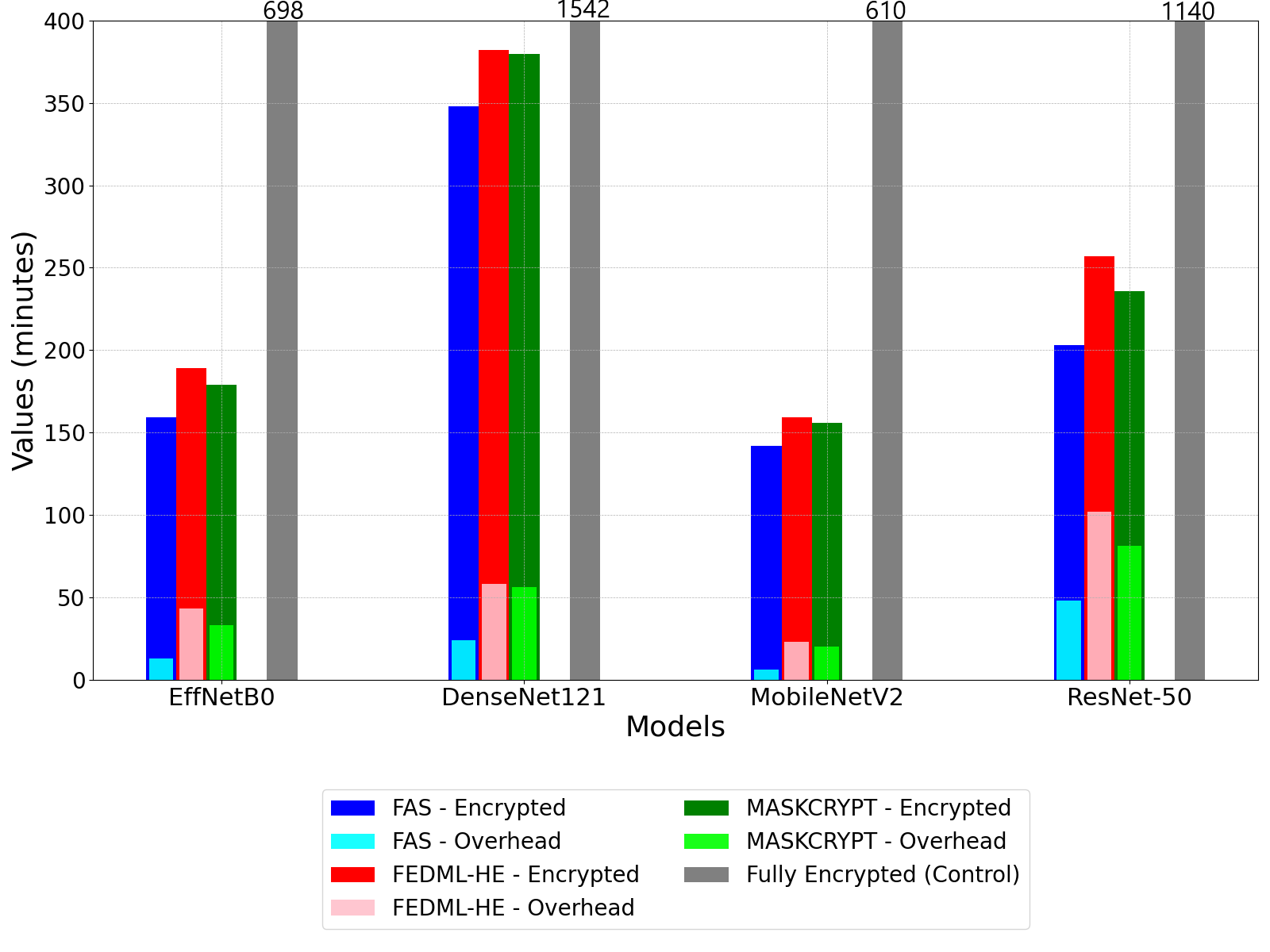}
\caption{Comparison of Partly Encrypted and Fully Encrypted Metrics Across Models (Diabetic Retinopathy).}
\label{fig_dia_comparison}
\end{figure}

\begin{figure}[!t]
\centering
\includegraphics[width=\linewidth]{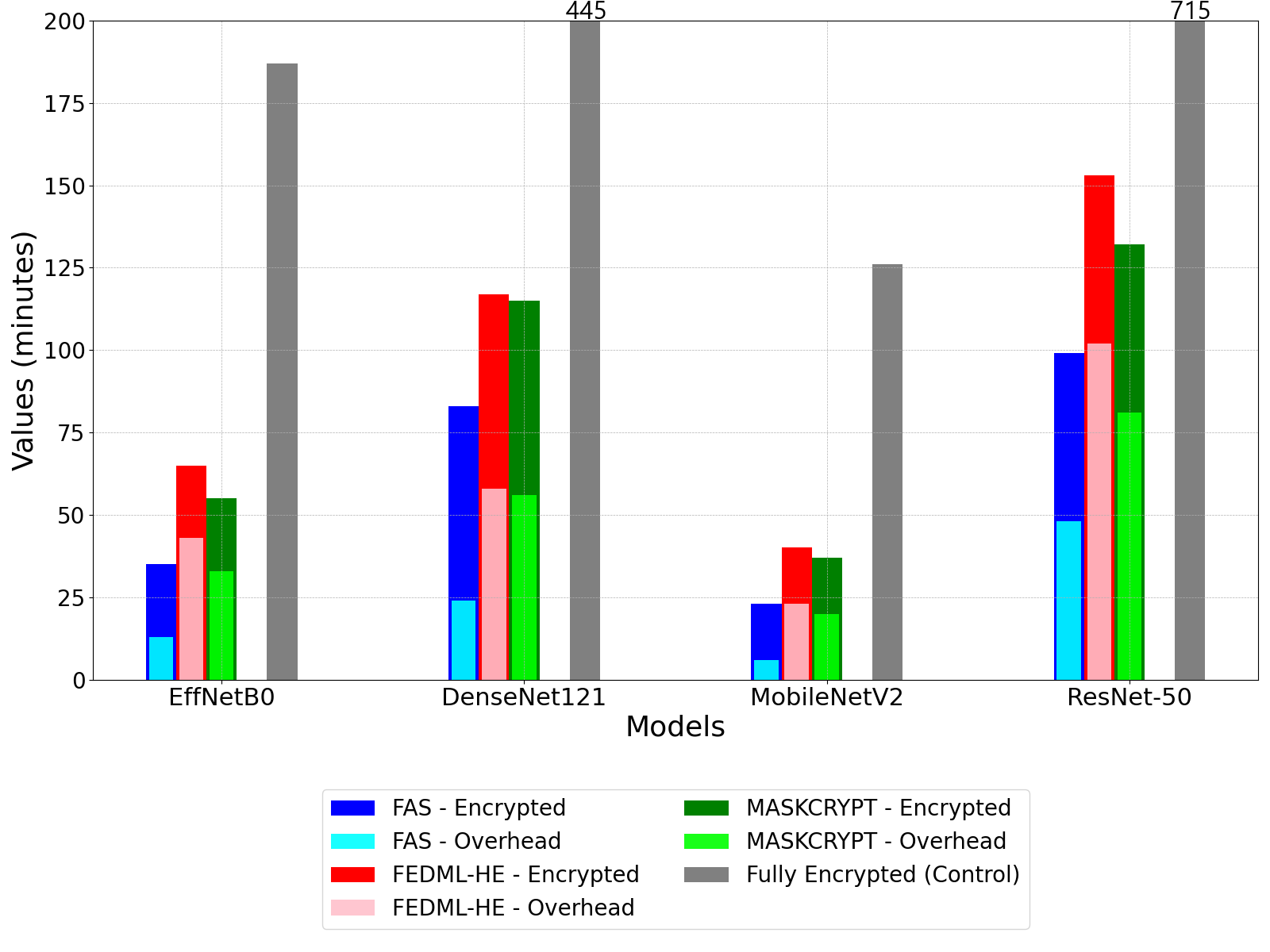}
\caption{Comparison of Partly Encrypted and Fully Encrypted Metrics Across Models (COVID).}
\label{fig_covid_comparison}
\end{figure}

\begin{figure}[!t]
\centering
\includegraphics[width=\linewidth]{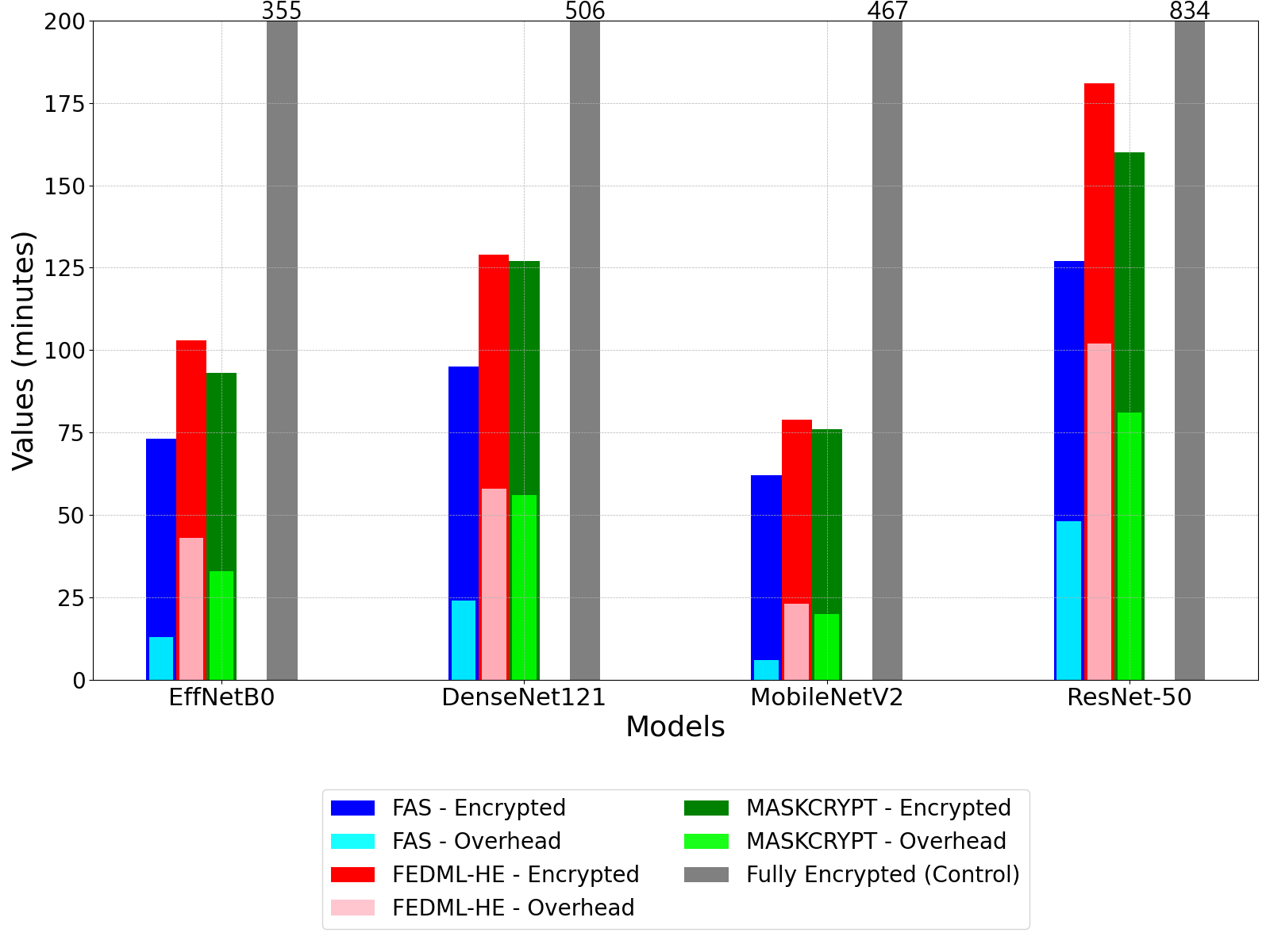}
\caption{Comparison of Partly Encrypted and Fully Encrypted Metrics Across Models (Lung).}
\label{fig_lung_comparison}
\end{figure}

\begin{figure}[!t]
\centering
\includegraphics[width=\linewidth]{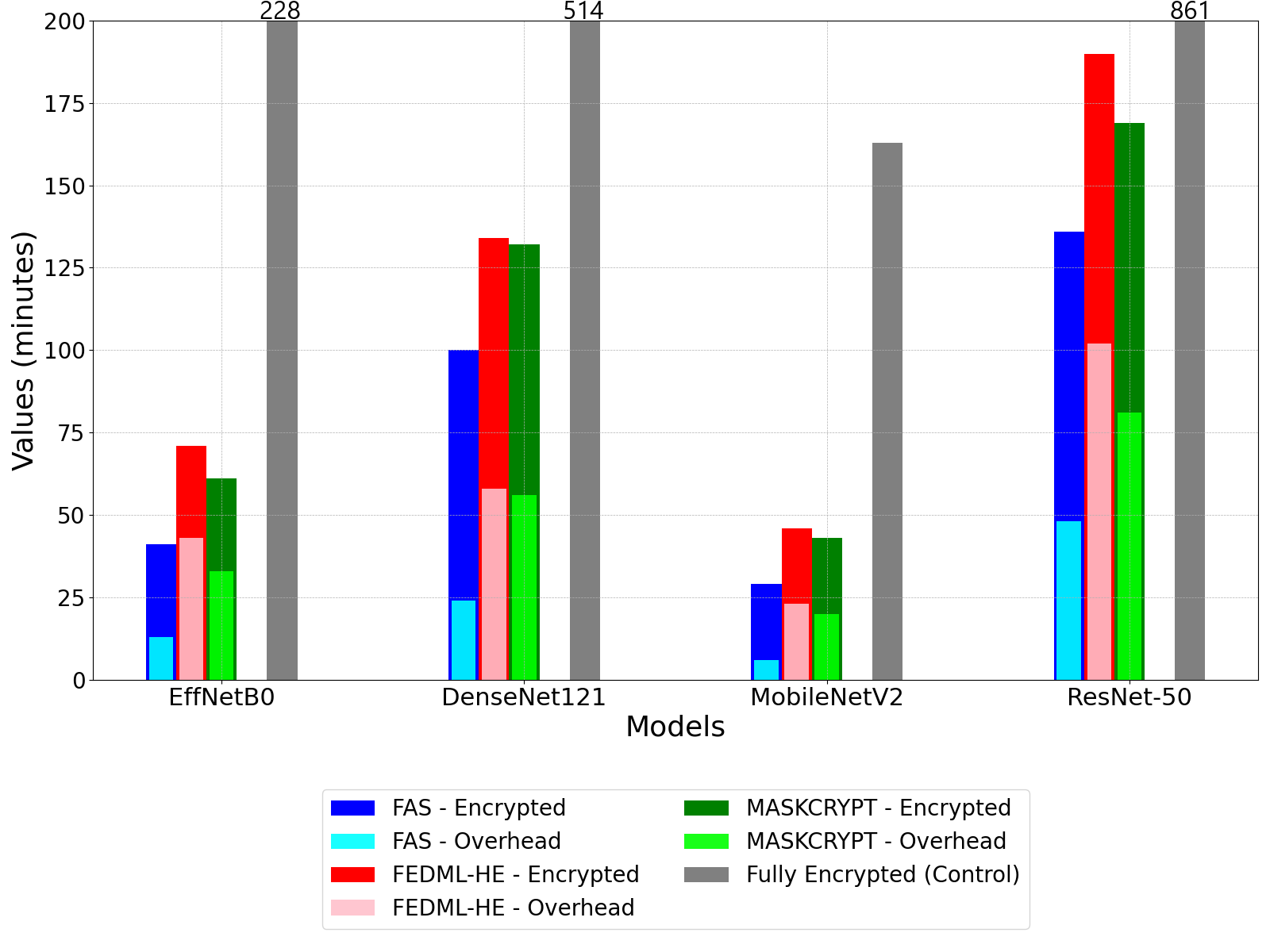}
\caption{Comparison of Partly Encrypted and Fully Encrypted Metrics Across Models (Kidney).}
\label{fig_kidney_comparison}
\end{figure}

Figures \ref{fig_cifar_comparison}, \ref{fig_covid_comparison}, \ref{fig_dia_comparison}, \ref{fig_kidney_comparison} and \ref{fig_lung_comparison} shows that across all datasets, FAS consistently outperforms FEDML-HE and MASKCRYPT in terms of total training time and overhead. The absence of sensitivity mask pretraining or extra per-round computations is the key factor contributing to its superior performance. These results establish FAS as the optimal choice for applications requiring fast and efficient encryption without compromising security. We can see that in smaller datasets which training time is not significant, difference between encryption technique becomes more obvious.
\setlength{\tabcolsep}{4pt} 
\renewcommand{\arraystretch}{1.2} 
\begin{table}[h]
    \centering
    \caption{Top 3 Cases Where FAS Outperforms Other Models}
    \begin{tabular}{|p{1.1cm}|p{1.8cm}|p{1.6cm}|p{0.9cm}|p{1.4cm}|}
        \hline
        \textbf{Dataset} & \textbf{Model} & \textbf{Compared} & \textbf{Enc. Impr. (\%)} & \textbf{Overhead Impr. (\%)} \\
        \hline
        COVID & MobileNetV2 & FEDML-HE & 42.50 & 73.91 \\
        COVID & EffNetB0 & FEDML-HE & 46.15 & 69.77 \\
        COVID & EffNetB0 & MASK-CRYPT & 36.36 & 60.61 \\
        \hline
    \end{tabular}
    \label{table-top3}
\end{table}

Table \ref{table-top3} presents the top three cases where the FAS encryption model demonstrates a higher efficiency compared to FEDML-HE and MASKCRYPT in terms of encryption time and overhead reduction.

COVID dataset with MobileNetV2: FAS achieves a 42.50\% faster encryption time and a 73.91\% lower overhead than FEDML-HE.
COVID dataset with EffNetB0: FAS provides a 46.15\% encryption improvement and a 69.77\% overhead reduction compared to FEDML-HE.
COVID dataset with EffNetB0 (vs. MASKCRYPT): FAS reduces encryption time by 36.36\% and overhead by 60.61\%.
These results suggest that FAS is particularly effective in reducing computational overhead while maintaining faster encryption speeds, especially in the COVID dataset with lightweight models.

\subsection{Handling Data Skew in Encryption-Based Techniques}

This section investigates how data skew affects different encryption techniques, with an emphasis on evaluating the robustness of the proposed FAS approach in comparison to existing methods like MASKCRYPT and Fedml-HE.

\subsubsection{Challenges with Mask-Based Techniques}
Both MASKCRYPT and Fedml-HE rely on sensitive masks to select important gradients for encryption.  that strongly affect how well the model performs. However, in scenarios with skewed data distributions, the importance of gradients becomes difficult to determine accurately. As a result, the encryption decisions made by these methods often degrade to levels resembling random encryption. This reduces their effectiveness, as observed in the MSSIM and VIFP score comparisons.

\subsubsection{Robustness of the FAS Technique}
In contrast, the proposed FAS technique, which operates as a random encryption method, is unaffected by data skew. Unlike mask-based techniques, FAS does not depend on the model's accuracy or gradient importance for its encryption process. Instead, it employs a combination of encryption, scrambling, and noise addition. These components ensure that the method maintains consistent performance irrespective of the underlying data distribution. The robustness of this approach makes it particularly suitable for scenarios where data skew is prevalent, such as FL environments with heterogeneous clients.

\subsubsection{Experimental Results and Comparisons}

\begin{table}[h]
    \centering
    \begin{tabular}{|p{1.3cm}|p{2cm}|p{1.7cm}|p{1.7cm}|}
        \hline
        \textbf{Dataset} & \textbf{Method} & \textbf{MSSIM (Skew/\newline Normal)} & \textbf{VIFP (Skew/\newline Normal)} \\
        \hline
        \multirow{4}{*}{Kidney} 
        & Fedml-HE & 63 / 57 & 18 / 13 \\
        & FAS & 62 / 61 & 16.5 / 15 \\
        & MASKCRYPT & 63 / 60 & 18 / 15 \\
        & Control & 70 / 70 & 20 / 20 \\
        \hline
        \multirow{4}{*}{Lung} 
        & Fedml-HE & 65 / 59 & 19 / 14 \\
        & FAS & 62 / 61 & 16.5 / 15 \\
        & MASKCRYPT & 66 / 62 & 19 / 16 \\
        & Control & 70 / 70 & 20 / 20 \\
        \hline
        \multirow{4}{*}{COVID} 
        & Fedml-HE & 66 / 60 & 20 / 15 \\
        & FAS & 62 / 61 & 16.5 / 15 \\
        & MASKCRYPT & 67 / 63 & 20 / 17 \\
        & Control & 70 / 70 & 20 / 20 \\
        \hline
        \multirow{4}{*}{\parbox{2cm}{Diabetic\\Retinopathy}} 
        & Fedml-HE & 64 / 58 & 18 / 14 \\
        & FAS & 62 / 61 & 16.5 / 15 \\
        & MASKCRYPT & 65 / 61 & 19 / 16 \\
        & Control & 70 / 70 & 20 / 20 \\
        \hline
        \multirow{4}{*}{CIFAR-10} 
        & Fedml-HE & 62 / 58 & 17 / 13 \\
        & FAS & 61 / 59 & 16.5 / 15 \\
        & MASKCRYPT & 64 / 59 & 18 / 14 \\
        & Control & 70 / 70 & 20 / 20 \\
        \hline
    \end{tabular}
    \caption{MSSIM and VIFP Scores for Skew and Normal Data across Datasets}
    \label{tab:results}
\end{table}

Table~\ref{tab:results} demonstrates that the proposed FAS technique, leveraging encryption, scrambling, and noise, exhibits superior robustness to data skew across all evaluated datasets—Kidney, Lung, COVID, and Diabetic Retinopathy—when compared to existing gradient-based encryption methods. This characteristic ensures that FAS maintains its performance and security advantages even in diverse, real-world scenarios with varying data distributions. Experimental results using the EfficientNetB0 model across these datasets show that the accuracy of FAS is minimally affected by data skew, while other methods experience significant performance degradation under similar conditions. These findings highlight the adaptability and reliability of the FAS technique in privacy-preserving machine learning applications.
\begin{itemize}
    \item For MASKCRYPT and Fedml-HE, the MSSIM and VIFP scores under skewed data are significantly degraded, closely resembling random encryption.
    \item In contrast, the FAS technique maintains reduced degragation across all conditions, as its random encryption approach is independent of the data distribution.
    \item The control value has been kept same with skewed and normal to show the random selective encryption comparison.
\end{itemize}

\subsubsection*{Effect of Data Skew on Training Performance}
\begin{figure}[!t]
    \centering
    \includegraphics[width=0.5\textwidth]{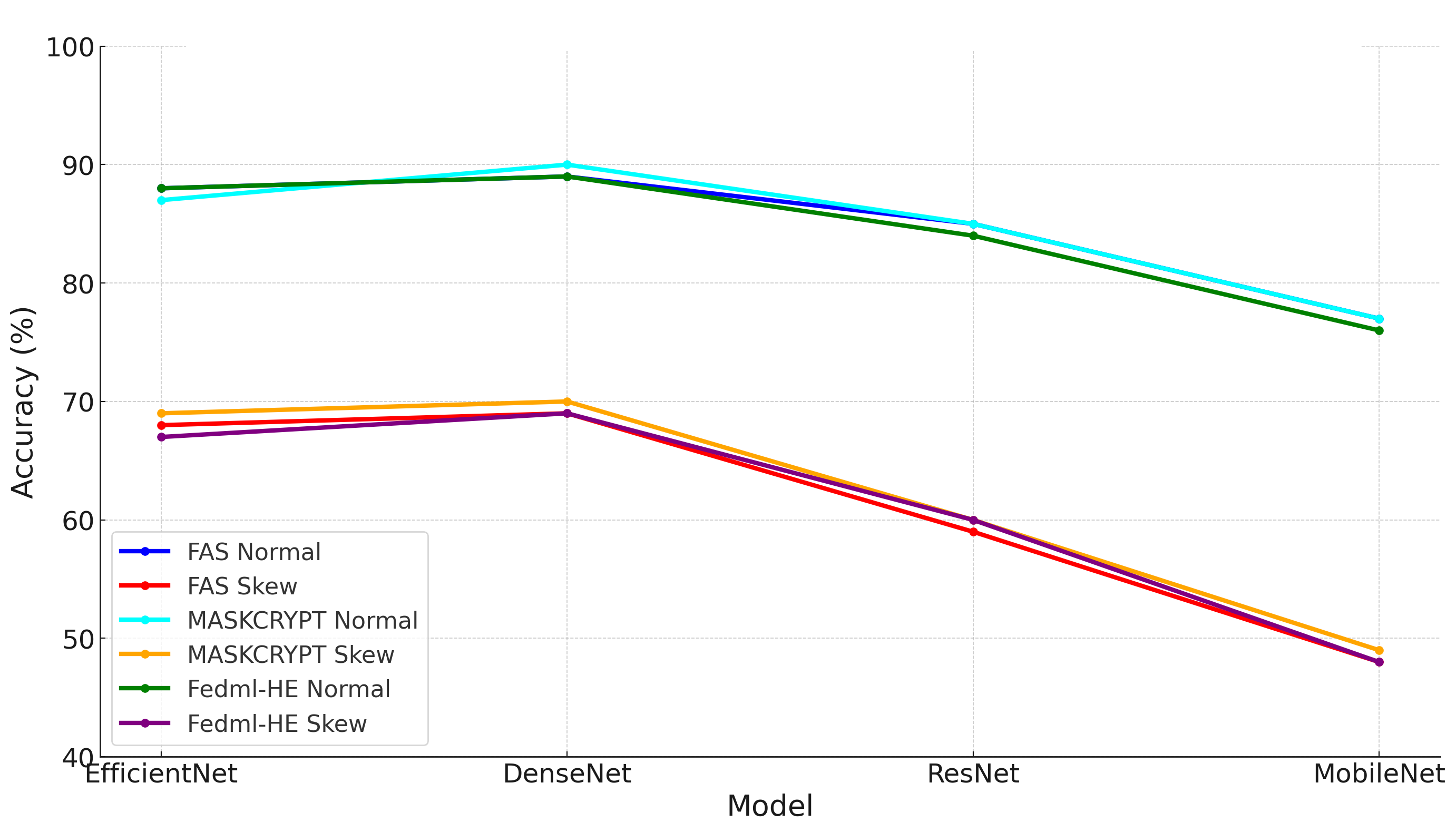}
    \caption{Effect of Data Skew on CIFAR Dataset Across Different Techniques}
    \label{fig:cifar_skew_graph}
\end{figure}
Figure \ref{fig:cifar_skew_graph} illustrates the uniform reduction in accuracy across three techniques—FAS, MASKCRYPT, and Fedml-HE—when training on skewed data using the CIFAR dataset. While all techniques show similar accuracy drops under skewed conditions, the implications for their security mechanisms differ significantly.

\begin{itemize}
    \item \textbf{MASKCRYPT and Fedml-HE:} Both techniques rely on sensitivity masks to ensure data security. As training accuracy decreases under skewed conditions, the quality and reliability of these sensitivity masks are compromised, potentially leaving the system insecure until the mask stabilizes.
    \item \textbf{FAS:} Unlike the other techniques, FAS does not depend on sensitivity masks. This independence ensures that the MS-SSIM and VIFP scores of FAS remain stable, even when training on skewed data. The inherent robustness of FAS allows it to maintain its security guarantees regardless of data distribution.
\end{itemize}

These results highlight a critical advantage of FAS: its ability to decouple data security from training accuracy. While longer training epochs can partially recover accuracy for MASKCRYPT and Fedml-HE, their reliance on sensitivity masks introduces a window of vulnerability during the calibration phase. In contrast, FAS maintains consistent security and performance, making it a more robust choice under challenging data conditions.

\section{Conclusion} 

This paper evaluated privacy-preserving techniques in FL, focusing on different datasets and comparing differential privacy, HE, and our custom FAS approach.

FHE offers the highest data protection but incurs significant computational costs, making it suitable only for scenarios prioritizing confidentiality over performance. Differential privacy provides lightweight privacy with minimal impact on computation, ideal for moderate security requirements. Our FAS method strikes a balance by achieving security comparable to full encryption while significantly reducing training time and overhead, making it efficient for large-scale FL and resource-constrained environments.

FAS's layered approach—combining selective encryption, bitwise scrambling, and differential noise—demonstrates strong resilience against model inversion attacks without requiring pre-training or complex mask aggregation, outperforming FedML-HE and MASKCRYPT in scalability and efficiency. FAS offers slighly better security compared to models that require accurate sensitivity masks for data skews or operate under general  low-accuracy conditions.

In summary, FAS offers an effective middle ground, balancing security and performance for real-time, privacy-sensitive applications like healthcare. Future research will refine these techniques and explore hybrid approaches across diverse datasets and federated environments to enhance scalability and applicability.

\section{Acknowledgments}
The first author (A. K.) was supported by the Republic of Türkiye.

\bibliographystyle{spmpsci}

\end{document}